\documentclass[a4paper,thm-restate,numberwithinsect]{lipics-v2021}

\usepackage{todonotes}
\usepackage{thmtools,thm-restate}

\newtheorem{open}{Open Problem}

\newcommand{\NP}{{\sf NP}}

\newcommand{\problem}[3]{
	\vspace{2mm}
	\noindent\fbox{
		\begin{minipage}{0.96\textwidth}
			\begin{tabular*}{\textwidth}{@{\extracolsep{\fill}}lr} #1 & \\ \end{tabular*}
			{\textbf{Input:}} #2 \\
			{\textbf{Question:}} #3
		\end{minipage}
	\vspace{2mm}
	}
}

\usepackage{regexpatch}
\makeatletter
\xpatchcmd\thmt@restatable{%
\csname #2\@xa\endcsname\ifx\@nx#1\@nx\else[{#1}]\fi
}{%
\ifthmt@thisistheone
\csname #2\@xa\endcsname\ifx\@nx#1\@nx\else[{#1}]\fi
\else
\csname #2\@xa\endcsname[{Restated}]
\fi}{}{}
\makeatother

\title{{E}dge {M}ultiway {C}ut and {N}ode {M}ultiway {C}ut Are Hard for Planar Subcubic Graphs\footnote{An extended abstract of this paper appeared in the proceedings of SWAT 2024~\cite{MPPSL24}.}} 
 
\titlerunning{EMWC and NMWC are hard for planar subcubic graphs}

\author{Matthew Johnson}{Durham University, Durham, United Kingdom}{matthew.johnson2@durham.ac.uk}{}{}
\author{Barnaby Martin}{Durham University, Durham, United Kingdom}{barnaby.d.martin@durham.ac.uk}{}{}
\author{Sukanya Pandey}{Utrecht University, Utrecht, The Netherlands}{s.pandey1@uu.nl}{0000-0001-5728-1120}{}
\author{Dani\"el Paulusma}{Durham University, Durham United Kingdom}{daniel.paulusma@durham.ac.uk}{0000-0001-5945-9287}{}
\author{Siani Smith}{University of Bristol and Heilbronn Institute for Mathematical Research, Bristol, United Kingdom}{siani.smith@bristol.ac.uk}{}{}
\author{Erik Jan van Leeuwen}{Utrecht University, Utrecht, The Netherlands}{e.j.vanleeuwen@uu.nl}{0000-0001-5240-7257}{}

\authorrunning{M. Johnson et al.}



\acknowledgements{The authors wish to thank Jelle Oostveen for his helpful comments.}

\nolinenumbers
\hideLIPIcs

\begin{document}
\maketitle

\begin{abstract}
It is known that the weighted version of {\sc Edge Multiway Cut} (also known as {\sc Multiterminal Cut}) is \NP-complete on planar graphs of maximum degree~$3$. In contrast, for the unweighted version, \NP-completeness is only known for planar graphs of maximum degree~$11$. In fact, the complexity of unweighted {\sc Edge Multiway Cut} was open for graphs of maximum degree~$3$ for over twenty years. We prove that the unweighted version is \NP-complete even for planar graphs of maximum degree~$3$. As  weighted {\sc Edge Multiway Cut} is polynomial-time solvable for graphs of maximum degree at most~$2$,  we have now closed the complexity gap.  We also prove that (unweighted) {\sc Node Multiway Cut}  (both with and without deletable terminals) is \NP-complete for planar graphs of maximum degree~$3$. By combining our results with known results, we can apply two meta-classifications on graph containment from the literature. This yields full dichotomies for all three problems on ${\cal H}$-topological-minor-free graphs and, should ${\cal H}$ be finite, on ${\cal H}$-subgraph-free graphs as well. 
Previously, such dichotomies were only implied for ${\cal H}$-minor-free graphs.
\end{abstract}

\section{Introduction}

In this paper we consider the unweighted edge and node versions of the classic {\sc Multiway Cut} problem, which is one of the most central separation/clustering graph problems with applications in, for example, computer vision~\cite{BirchfieldT99,BoykovVZ98} and multi-processor scheduling~\cite{Stone77}.

To define these problems, let $G=(V,E)$ be a graph. For a subset $S$ of either vertices or edges of $G$, let $G-S$ denote the graph obtained from $G$ after deleting all elements, either vertices (and incident edges) or edges, of $S$.
Now, let $T\subseteq V$ be a set of specified vertices that are called the {\it terminals} of $G$. A set $S\subseteq E$ is an {\it edge multiway cut} for $(G,T)$ if every connected component of $G-S$ contains at most one vertex of~$T$. In order words, removing $S$ pairwise disconnects the terminals of~$T$; see Figure~\ref{fig:examples} for an example.
We define the notion of a {\it node multiway cut} $S\subseteq V$ in the same way, but there are two versions depending on whether or not
 $S$ can contain vertices of~$T$; see again Figure~\ref{fig:examples}.
 This leads to the following three decision problems, 
where the second one is also known as {\sc Unrestricted Node Multiway Cut} and the third one as {\sc Restricted Node Multiway Cut} or {\sc Node Multiway Cut with Undeletable Terminals}. 
\begin{figure}[t]
	\centering
\begin{minipage}{0.45\textwidth}
		\raggedright
		\includegraphics[scale=0.7]{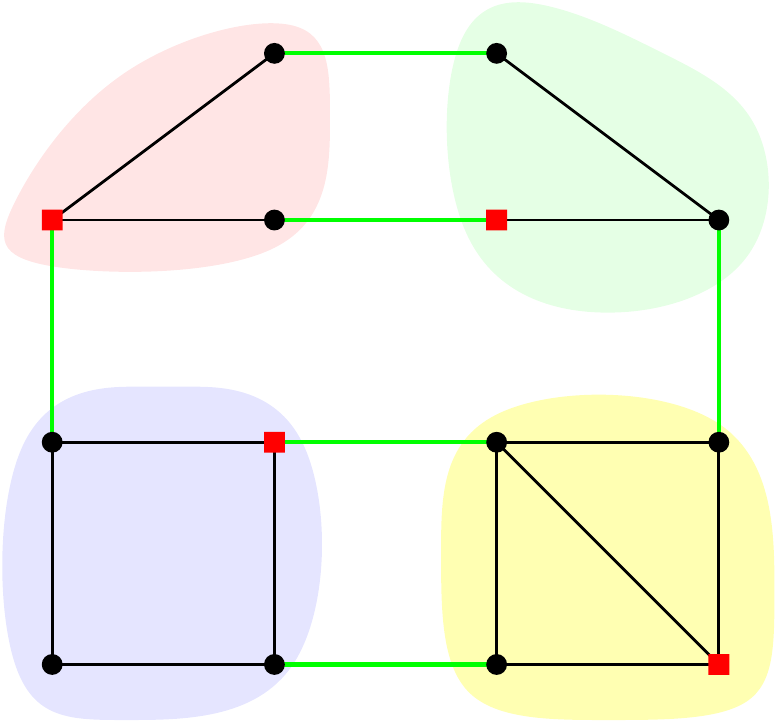}	
\end{minipage}
\begin{minipage}{0.45\textwidth}
	\raggedleft
	\includegraphics[scale=0.7]{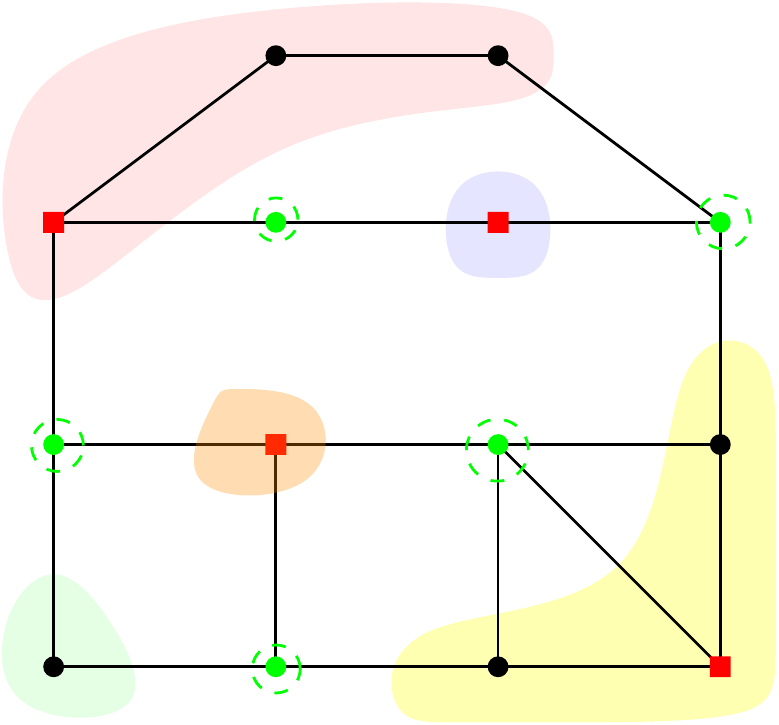}
\end{minipage}
\centering
\includegraphics[scale=0.7]{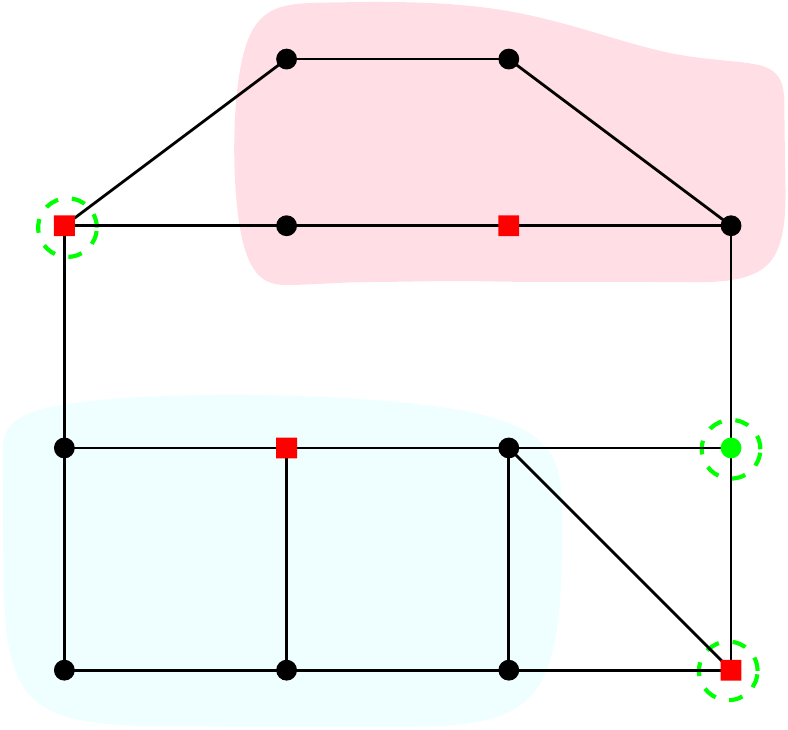}
\caption{The three different types of multiway cuts that we consider in our paper. In all figures, the red square nodes form the terminal set~$T$. In the top left figure, the green lines form an edge multiway cut. In the top right, the green encircled vertices form a node multiway cut not containing a vertex of $T$. In the bottom figure, the green encircled vertices form a node multiway cut that contains two vertices of $T$. The coloured parts depict the components formed after removing the edges/vertices of the multiway cut.}
\label{fig:examples}
\end{figure}

\medskip
\problem{\sc Edge Multiway Cut}{A graph $G$, a set of terminals $T\subseteq V$ and an integer $k$.}{Does $(G,T)$ have an edge multiway cut $S\subseteq E$ of size at most $k$?}

\medskip
\problem{\sc Node Multiway Cut with Deletable Terminals}{A graph $G$, a set of terminals $T\subseteq V$ and an integer $k$.}{Does $(G,T)$ have a node multiway cut $S\subseteq V$ of size at most $k$?}

\medskip
\problem{\sc Node Multiway Cut}{A graph $G$, a set of terminals $T\subseteq V$ and an integer $k$.}{Does $(G,T)$ have a node multiway cut $S\subseteq V\setminus T$ of size at most $k$?}

\medskip
\noindent
In {\sc Weighted Edge Multiway Cut}, we are given a function $\omega: E(G)\rightarrow \mathbb{Q}^+$. The goal is to decide if $(G,T)$ admits an edge multiway cut of total weight at most $k$. If $\omega\equiv 1$, then we obtain {\sc Edge Multiway Cut}. 
Similarly, we can define weighted variants of both versions of {\sc Node Multiway Cut} with respect to a node weight function $\omega: V(G)\rightarrow \mathbb{Q}^+$.

The above problems have been studied extensively; see, for example,~\cite{BergougnouxPT22, CKR00, CCF14, CLL09, CHM13, CPPW13, GalbyMSST22, GargVY04, Ha98, KM12, Ma12, PV22,PapadopoulosT20}. The problems can be thought of as the natural dual problems of the {\sc Steiner Tree} problem. 
In their famous study of {\sc Edge Multiway Cut}, Dahlhaus et al.~\cite{DJPSY94} showed that it is \NP-complete even if the set of terminals has size $|T|=3$. Garg et al.~\cite{GargVY04} showed the same for {\sc Node Multiway Cut}. 
We note that this is a tight result: if $|T|=2$, then both problems reduce to the {\sc Minimum Cut} problem. The latter problem can be modelled as a maximum flow problem, and hence is well known to be solvable in polynomial time~\cite{FF56}.
Note that {\sc Node Multiway Cut with Deletable Terminals} is trivially polynomial-time solvable for any fixed $|T|$.

\subparagraph*{Our Focus}
 A graph is {\it subcubic} if it has maximum degree at most~$3$. Our goal in this paper is to answer the following question: 

\medskip
\noindent
{\it What is the computational complexity of {\sc Edge Multiway Cut} and both versions of {\sc Node Multiway Cut} for planar subcubic graphs?} 

\subparagraph{Motivation} 
Our first reason for considering the above research question is due to a complexity gap that was left open in the literature for over twenty years.  
That is, 
in addition to their \NP-completeness result for $|T|=3$,
Dahlhaus et al.~\cite{DJPSY94} also proved that {\sc Weighted Edge Multiway Cut} is \NP-complete on planar subcubic graphs using integral edge weights. 
Any edge of integer weight $j$ can be replaced by $j$ parallel edges (and vice versa) without changing the problem. Hence, their reduction implies that {\sc Edge Multiway Cut} is \NP-complete on planar graphs of maximum degree at most~$11$~ \cite[Theorem~2b]{DJPSY94}. 
Dahlhaus et al.~\cite{DJPSY94} write that 

\medskip
\noindent
``{\it The degree bound of~$11$ is not the best possible. Using a slight variant on the construction and considerably more complicated arguments, we believe it can be reduced at least to~$6$}'', 

\medskip
\noindent
but no further arguments were given. 
Even without the planarity condition and only focussing on the maximum degree bound, the hardness result of 
Dahlhaus et al.~\cite{DJPSY94} is still best known. 
Given that the problem is polynomial-time solvable if the maximum degree is~$2$, this means that there is a significant complexity gap that has yet to be addressed.

To the best of our knowledge, there is no explicit hardness result in the literature that proves \NP-completeness of either version of {\sc Node Multiway Cut} on graphs of any fixed degree or on planar graphs. 
However, known and straightforward reductions (see e.g.~\cite{GargVY04,PapadopoulosT20}) immediately yield \NP-hardness on planar subcubic graphs for {\sc Node Multiway Cut with Deletable Terminals} (see Theorem~\ref{thm:NMwCDT:C2}), but only on planar graphs of maximum degree~$4$ for {\sc Node Multiway Cut} (see Proposition~\ref{prp:NMWC4}). 

Our second reason is the central role planar subcubic graphs play in complexity dichotomies of graph problems restricted to graphs that do not contain any graph from a set~${\cal H}$ as a topological minor\footnote{A graph $G$ contains a graph $H$ as a {\it topological minor} if $G$ can be modified into $H$ by a sequence of edge deletions, vertex deletions and vertex dissolutions, where a vertex dissolution is the contraction of an edge incident to a vertex of degree~$2$ whose (two) neighbours are non-adjacent.}
 or subgraph; such graphs are said to be {\it ${\cal H}$-topological-minor-free} or {\it ${\cal H}$-subgraph-free}, respectively.
 For both the topological minor containment relation~\cite{RS86} and the subgraph relation (see~\cite{JMOPPSV}) meta-classifications exist.
 To apply these meta-classifications, a problem must satisfy certain conditions, in particular being \NP-complete for subcubic planar graphs for the topological minor relation, and being \NP-complete for subcubic graphs for the subgraph relation. These two conditions are {\it exactly what is left to prove}  
for {\sc Edge Multiway Cut} and both versions of {\sc Node Multiway Cut}.
In contrast, the results of~\cite{ALS91,DJPSY94,RS86} and the aforementioned reductions from~\cite{GargVY04,PapadopoulosT20}  imply that all three problems are fully classified for ${\cal H}$-minor-free graphs: the problems are polynomial-time solvable if ${\cal H}$ contains a planar graph and \NP-complete otherwise (see also~\cite{JMOPPSV}).
Hence, determining the complexity status of our three problems for planar subcubic graphs is a pressing issue 
for obtaining new complexity classifications for ${\cal H}$-topological-minor-free graphs and ${\cal H}$-subgraph-free-graphs.

Our third reason is the rich tradition to investigate the \NP-completeness of problems on subcubic graphs and planar subcubic  graphs (see e.g.~the list given by 
Johnson et al.~\cite{JMOPPSV}) which continues till this day, as evidenced by recent \NP-completeness results for subcubic graphs (e.g.~\cite{BBJKLMOPS23,MP22}) and planar subcubic graphs (e.g.~\cite{BCD23,ZZ20}).
We also note 
that {\sc Edge Multicut}, the standard generalization of {\sc Edge Multiway Cut} where given pairs of terminals must be disconnected by an edge cut', is \NP-complete even on subcubic trees~\cite{CF01}. 

For the above reasons,
resolving the complexity status of our three problems restricted to (planar) subcubic graphs is long overdue.

\subsection{Our Results} 

The following three results fully answer our research question. 

\begin{restatable}{theorem}{EMWCthm}\label{thm:MwC:C2}
	{\sc Edge Multiway Cut} is \NP-complete for planar subcubic graphs.
\end{restatable}

\begin{restatable}{theorem}{NMWCthm}\label{thm:NMwC:C2}
	{\sc Node Multiway Cut} is \NP-complete for planar subcubic graphs.
\end{restatable}

\begin{restatable}{theorem}{NMWCDTthm}\label{thm:NMwCDT:C2}
	{\sc Node Multiway Cut with Deletable Terminals} is \NP-complete for planar subcubic graphs.
\end{restatable}

\noindent
We prove Theorem~\ref{thm:MwC:C2} in Section~\ref{s-1}; Theorem~\ref{thm:NMwC:C2} in Section~\ref{s-2}; and Theorems~\ref{thm:NMwCDT:C2} in Section~\ref{s-3}. 

In spirit, our construction for {\sc Edge Multiway Cut} in Theorem~\ref{thm:MwC:C2} is similar to the one by Dahlhaus et al.~\cite{DJPSY94} for graphs of maximum degree~$11$. For non-terminal vertices of high degree, a local replacement by a (sub)cubic graph is relatively easy. However, for terminal vertices of high degree, a local replacement strategy seems impossible. Hence, the fact that terminals in the construction of Dahlhaus et al.~\cite{DJPSY94} can have degree up to~$6$ becomes a crucial bottleneck. 
To ensure that our constructed graph has maximum degree~$3$, we therefore need to build different gadgets. We then leverage several deep structural properties of the edge multiway cut in the resulting instance, making for a significantly more involved and technical correctness proof.
Crucially, we first prove \NP-completeness for a weighted version of the problem on graphs of maximum degree~$5$, in which 
each terminal is incident with exactly one edge of weight~$3$.
In the final step of our construction, we replace weighted edges and high-degree vertices with appropriate gadgets. 

The \NP-hardness for {\sc Node Multiway Cut} for planar subcubic graphs shown in Theorem~\ref{thm:NMwC:C2} follows from the \NP-hardness of {\sc Edge Multiway Cut} by constructing the line graph of input graph.

The \NP-hardness for {\sc Node Multiway Cut with Deletable Terminals} on planar subcubic graphs shown in Theorem~\ref{thm:NMwCDT:C2} follows from a straightforward reduction from {\sc Vertex Cover}.

\subsection{Consequences} 

As discussed above, we immediately have the following dichotomy.

\begin{corollary}
For every $\Delta \geq 1$, {\sc Edge Multiway Cut} and both versions of {\sc Node Multiway Cut} on graphs of maximum degree~$\Delta$ are polynomial-time solvable if $\Delta\leq 2$, and \NP-complete if $\Delta\geq 3$.
\end{corollary}

\noindent
From a result of~Robertson and Seymour~\cite{RS86}, it follows that any problem~$\Pi$ that is \NP-hard on subcubic planar graphs but polynomial-time solvable for graphs of bounded treewidth can be fully classified on ${\cal H}$-topological minor-free graphs. Namely, $\Pi$ is polynomial-time solvable if ${\cal H}$ contains a subcubic planar graph and \NP-hard otherwise.
It is known that {\sc Edge Multiway Cut} and both versions of {\sc Node Multiway Cut} satisfy the second property~\cite{ALS91}. As Theorems~\ref{thm:MwC:C2}--\ref{thm:NMwC:C2} show the first property, we obtain the following dichotomy.

\begin{corollary}
For every set of graphs $\mathcal{H}$, {\sc Edge Multiway Cut} and both versions of {\sc Node Multiway Cut} on $\mathcal{H}$-topological-minor-free graphs are polynomial-time solvable if $\mathcal{H}$ contains a planar subcubic graph, and \NP-complete otherwise.
\end{corollary}

\noindent
Let 
the {\it $\ell$-subdivision} of a graph~$G$ be the graph obtained from~$G$ after replacing each edge $uv$ by a path of 
$\ell+1$ edges with end-vertices $u$ and $v$. 
A problem~$\Pi$ is \NP-hard
{\it under edge subdivision of subcubic graphs} if for every integer $j \geq 1$ there is an~$\ell \geq j$ such that:
if $\Pi$ is \NP-hard for the class ${\cal G}$ of subcubic graphs, then $\Pi$ is \NP-hard for the class $G^\ell$ consisting of the $\ell$-subdivisions of the graphs in ${\cal G}$.
Now say that $\Pi$ is polynomial-time solvable on graphs of bounded treewidth and \NP-hard for subcubic graphs and under edge subdivision of subcubic graphs. The meta-classification from 
Johnson et al.~\cite{JMOPPSV} states that for every {\it finite} set ${\cal H}$, $\Pi$ on ${\cal H}$-subgraph-free graphs is polynomial-time solvable if ${\cal H}$ contains a graph from ${\cal S}$, and \NP-hard otherwise. Here, ${\cal S}$ is the set consisting of all disjoint unions of zero or more paths and subdivided claws ($4$-vertex stars in which edges may be subdivided). Figure~\ref{fig:s} shows an example of a graph belonging to ${\cal S}$. Results from
Arnborg, Lagergren and Seese~\cite{ALS91} and Johnson et al.~\cite{JMOPPSV} 
show the first two properties. Theorems~\ref{thm:MwC:C2}--\ref{thm:NMwC:C2} show the last property. Thus,~we obtain: 

\begin{figure}[t]
	\centering
	\includegraphics[width=0.6\textwidth]{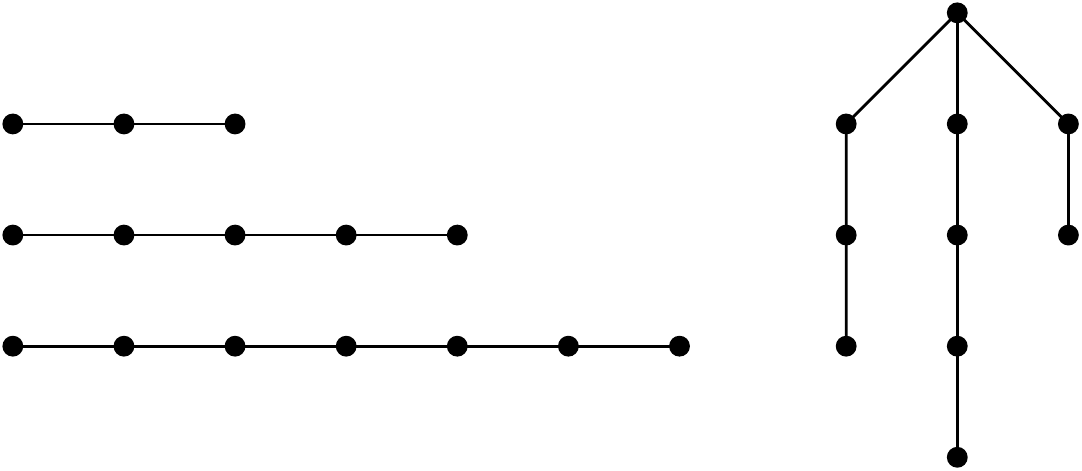}
	\caption{An example of a graph, namely $P_1+P_5+P_7+S_{2,3,4}$, that belongs to the set ${\cal S}$.}\label{fig:s}
\end{figure}

\begin{corollary}
For every finite set of graphs $\mathcal{H}$, {\sc Edge Multiway Cut} and both versions of {\sc Node Multiway Cut} on $\mathcal{H}$-subgraph-free graphs are polynomial-time solvable if $\mathcal{H}$ contains a graph from $\mathcal{S}$, and \NP-complete  otherwise.
\end{corollary}

\section{The Proof of Theorem~\ref{thm:MwC:C2}}\label{s-1}

In this section, we show that {\sc Edge Multiway Cut} is \NP-complete on subcubic graphs. We reduce the problem from {\sc Planar 2P1N-3SAT}, which is a restricted version of {\sc 3-SAT}. Given a CNF-formula $\Phi$ with the set of variables $X$ and the set of clauses $C$, the \emph{incidence graph} of the formula is the graph $G_{X,C}$ which is a bipartite graph with one of the partitions containing a vertex for each variable and the other partition containing a vertex for each clause of $\Phi$. There exists in $G_{X, C}$ an edge between a variable-vertex and a clause-vertex if and only if the variable appears in the clause. We define {\sc Planar 2P1N-3SAT} as follows.

\problem{{\sc Planar 2P1N-3SAT}}{A set $X= \{x_1, \ldots, x_n\}$ of variables and a CNF formula $\Phi$ over $X$ and clause set $C$ with each clause containing at most three literals and each variable occurring twice positively and once negatively in $\Phi$ such that $G_{X, C}$ is planar.}{Is there an assignment $\mathcal{A}: X \rightarrow \{0,1\}$ that  satisfies $\Phi$?}

\medskip
\noindent
The above problem was shown to be \NP-complete by 
Dahlhaus et al.~\cite{DJPSY94}. By their construction, each variable occurs in at least two clauses having size~$2$. This property becomes important later in our \NP-completeness proof.

We need two further definitions. Recall that in {\sc Weighted Edge Multiway Cut}, we are given a function $\omega: E(G)\rightarrow \mathbb{Q}^+$ in addition to $G,T,k$. The goal is to decide if $(G,T)$ admits an edge multiway cut of total weight at most $k$. If the image of $\omega$ is the set $X$, we denote the corresponding {\sc Weighted Edge Multiway Cut} problem as $X$-{\sc Edge Multiway Cut}. Also note that if an edge/node multiway cut $S$ has smallest possible size (weight) among all edge/node multiway cuts for the pair $(G,T)$, then $S$ is a \emph{minimum(-weight)} edge/node multiway cut.

We show the reduction in two steps. In the first step, we reduce from {\sc Planar 2P1N-3SAT} to {\sc $\{1,2,3,6\}$-Edge Multiway Cut} restricted to planar graphs of maximum degree~$5$ where the terminals all have degree~$3$. In the second step, we show how to make the instance unweighted while keeping it planar and making its maximum degree bounded above by~$3$.

\EMWCthm*

\begin{proof}
	Clearly, {\sc Edge Multiway Cut} is in \NP.
	We reduce {\sc Edge Multiway Cut} from {\sc Planar 2P1N-3SAT}. Let $\Phi$ be a given CNF formula with at most three literals in each clause and each variable occurring twice positively and once negatively. 
	
	We assume that each clause has size at least~$2$ and every variable occurs in at least two clauses of size~$2$. Let $X = \{x_i \mid 1\leq i \leq n\}$ be the set of variables in $\Phi$ and $C = \{c_j \mid 1\leq j \leq m\}$ be the set of clauses. We assume that the incidence graph $G_{X,C}$ is planar. By the reduction 
of Dahlhaus et al.~\cite{DJPSY94}, {\sc Planar 2P1N-3SAT} is \NP-complete for such instances.

	\begin{figure}
		\centering
		\includegraphics[width=1\linewidth]{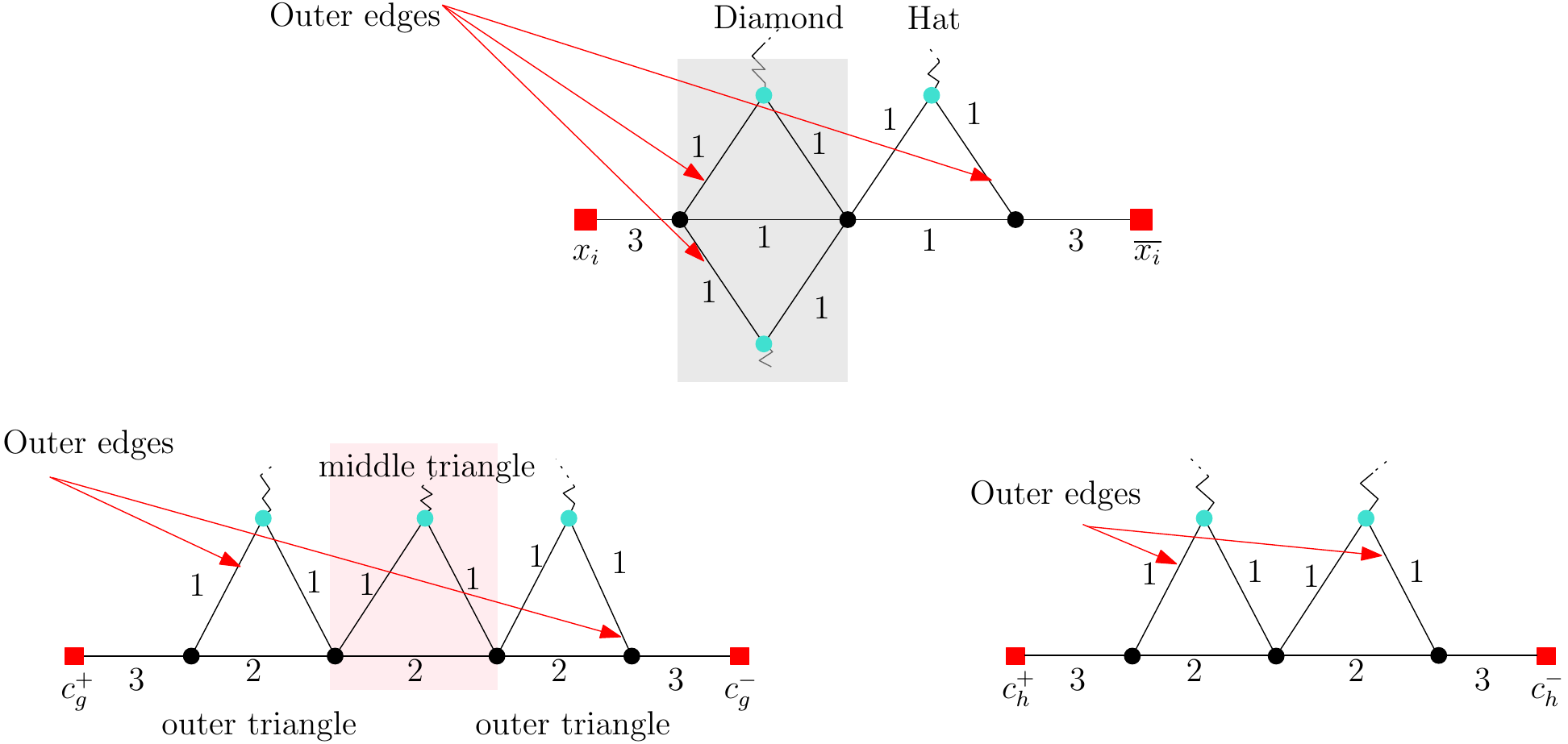}
	
\caption{The gadgets for the variables (top) as well as those for the clauses (bottom). The bottom-left gadget corresponds to a clause with three literals whereas the bottom-right one corresponds to a clause with two literals. The terminals are depicted as red squares.}\label{fig:gadgets}	
	\end{figure}

	We now describe the graph construction.	For each vertex of $G_{X,C}$ corresponding to a clause $c_j$ in $C$, we create a clause gadget (depending on the size of the clause), as in Figure~\ref{fig:gadgets}. For each vertex of $G_{X,C}$ corresponding to a variable $x_i \in X$, we create a variable gadget, also shown in Figure~\ref{fig:gadgets}. The gadgets have two terminals each (marked as red squares in Figure~\ref{fig:gadgets}), a positive and a negative one. In a variable gadget, the positive terminal is attached to the diamond and the negative one to the hat, by edges of weight~$3$; refer to Figure~\ref{fig:gadgets}. In a clause gadget, each literal corresponds to a triangle, with these triangles connected in sequence, and the positive and negative terminal are attached to triangles at the start and end of the sequence, again by edges of weight~$3$.

	Each degree-$2$ vertex in a gadget (marked blue in Figure~\ref{fig:gadgets}) is called a \emph{link}. The two edges incident on a link are called \emph{connector-edges}. The edge of such a triangle that is not incident on the link is called the \emph{base} of the triangle. For a variable $x_i$, if $x_i \in c_j$ and $x_i \in c_k$ for clauses $c_j,c_k$, then we connect the links of the diamond of $x_i$ to some link of the gadgets for $c_j$ and $c_k$, each by identifying them with the respective link in the clause gadget. If $\overline{x_i} \in c_l$ for clause $c_l$, then we connect the link of the hat of $x_i$ and some link on the gadget for $c_l$, again by identifying the two links. An example of such variable and clause connections is depicted in Figure~\ref{fig:connection}. The structure formed by the link and the four connector-edges incident on it, is referred to as a \emph{link-structure}. By the assumptions on $\Phi$, we can create the link-structures such that each link in the variable gadget participates in exactly one link-structure and corresponds to one occurrence of the variable. Similarly, each link of a clause gadget participates in exactly one link-structure.
	
	The graph thus created is denoted by $G$. We can construct $G$ in such a way that it is planar, because $G_{X,C}$ is planar and has maximum degree~$3$. Note that $G$ has maximum degree~$5$. Let $T$ be the set of terminals in the constructed graph $G$. Note that $G$ has a total of $2n+2m$ terminals.
	
	We observe that all edges in $G$ have weight at most~$6$. Non-terminal vertices are incident on edges of total weight at most~$8$. Crucially, terminals are incident on edges of total weight at most~$3$.

	\begin{figure}[h]
		\begin{minipage}[t]{\textwidth}
			\centering
			\includegraphics[width=0.45\linewidth]{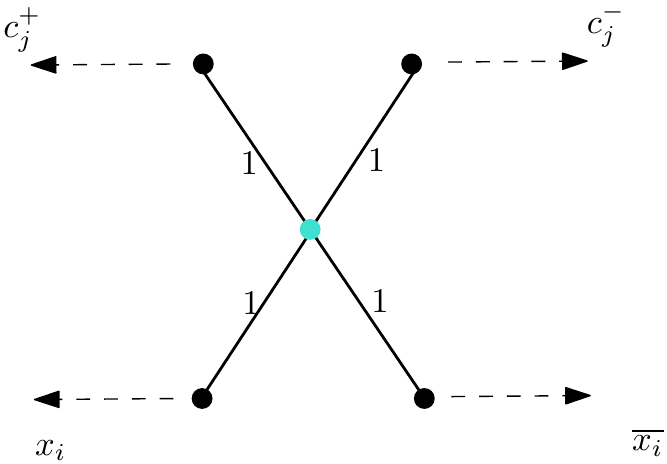}
		\end{minipage}		
		\begin{minipage}[t]{\textwidth}
			\centering
			
			\caption{The figure shows a link-structure formed by the connector-edges of a clause-triangle and its corresponding variable-triangle. The two bases that complete the triangles are not drawn.}\label{fig:link}		
		\end{minipage}
	\end{figure}
	
	We introduce some extra notions to describe the constructed graph $G$. The connector-edges closest to the terminals are called \emph{outer edges}, as indicated in Figure~\ref{fig:gadgets}. The structure formed by the two pairs of connector-edges and the link is called the \emph{link-structure}; see Figure~\ref{fig:link}. Since each variable occurs twice positively and once negatively in $\Phi$, the constructed graph $G$ has exactly $3n$ link-structures.

	We now continue the reduction to obtain an unweighted planar subcubic graph.
	We replace all the edges in $G$ of weight greater than $1$ by as many parallel edges between their end-vertices as the weight of the edge. Each of these parallel edges has weight~$1$. We refer to this graph as $G'$. Next, for each vertex $v$ in $G'$ of degree greater than $3$, we replace $v$ by a large honeycomb (hexagonal grid), as depicted in Figure~\ref{fig:honeycomb}, of $1000 \times 1000$ cells (these numbers are picked for convenience and not optimized). The neighbours of $v$, of which there are at most six by the construction of $G$, are now attached to distinct degree-$2$ vertices on the boundary of the honeycomb such that the distance along the boundary between any pair of them is $100$ cells of the honeycomb. These degree-$2$ vertices on the boundary are called the \emph{attachment points} of the honeycomb. The edges not belonging to the honeycomb that are incident on these attachment points are called \emph{attaching edges}. In the construction, we ensure that the attaching edges occur in the same cyclical order on the boundary as the edges to the neighbours of $v$ originally occurred  around $v$. Let the resultant graph be $\tilde{G}$.
	
	\begin{figure}[t]
		\centering
		\includegraphics[width=0.95\linewidth]{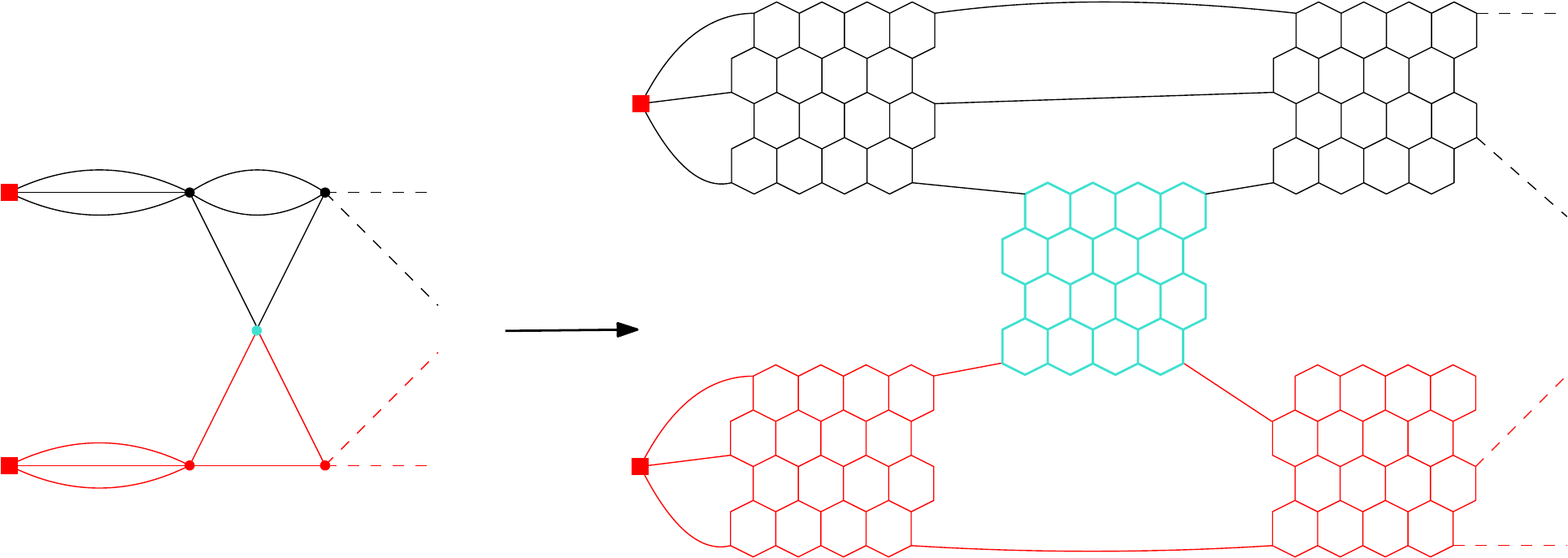}	
		\caption{Construction of $\tilde{G}$ from $G$ by replacing every edge of weight greater than 1 by as many parallel edges as its weight and then replacing the vertices of degree greater than 3 by a honeycomb of size $1000 \times 1000$.}
		\label{fig:honeycomb}
	\end{figure}
	
	Note that the degree of any vertex in $\tilde{G}$ is at most~$3$. For terminals, this was already the case in $G'$. Note that, therefore, terminals were not replaced by honeycombs to obtain~$\tilde{G}$. For non-terminals, this is clear from the construction of $G'$ and $\tilde{G}$. Moreover, all the edge weights of $\tilde{G}$ are equal to~$1$, and thus we can consider it unweighted. Also, all the replacements can be done as to retain a planar embedding of $G$ and hence, $\tilde{G}$ is planar. $\tilde{G}$ has size bounded by a polynomial in $n+m$ and can be constructed in polynomial time. Finally, we set $k=7n+2m$.
		
	For the sake of simplicity, we shall first argue that $\Phi$ is a {\sc yes} instance of {\sc Planar 2P1N-3SAT} if and only if $(G,T,k)$ is a {\sc yes} instance of {\sc $\{1,2,3,6\}$-Edge Multiway Cut}. Later, we show that the same holds for $\tilde{G}$ by proving that no edge of any of the honeycombs is ever present in any minimum edge multiway cut in $\tilde{G}$. 
	
	Suppose that $\mathcal{A}$ is a truth assignment satisfying $\Phi$. Then, we create a set of edges $S \subseteq E(G)$, as follows: 
	
	\begin{itemize}
		\item If a variable is set to ``true'' by $\mathcal{A}$, then add to $S$ all the three edges of the hat in the corresponding gadget. If a variable is set to ``false'' by $\mathcal{A}$, then add to $S$ all the five edges of the diamond.
		\item For each clause, pick a true literal in it and add to $S$ all the three edges of the clause-triangle corresponding to this literal.
		\item Finally, for each link-structure with none of its edges in $S$ yet, add the two connector-edges of its clause-triangle to $S$.
	\end{itemize}
	
	\begin{claim} \label{clm:forward}
		$S$ is an edge multiway cut of $(G, T)$ of weight at most $7n+2m$.
	\end{claim}
	\begin{claimproof}
		For each variable, either the positive literal is true, or the negative one. Hence, either all the three edges of its hat are in $S$ or all the five edges of the diamond. Therefore, all the paths between terminal pairs of the form  $x_i - \overline{x_i}$, for all $1\leq i \leq n $, are disconnected in $G - S$. Consider the link-structure in Figure~\ref{fig:link}. By our choice of $S$, at least one endpoint of each link in $G - S$ is a vertex of degree~$1$, hence a dead end. Therefore, no path connecting any terminal pair in $G - S$ passes through any link.  As all the paths in $G$ between a variable-terminal and a clause-terminal must pass through some link, we know that all terminal pairs of this type are disconnected in $G - S$. Since $\mathcal{A}$ is a satisfying truth assignment of $\Phi$, all the edges of one triangle from every clause gadget are in $S$. Hence, all the paths between terminal pairs of the form $c_j^+ - c_j^-$, for all $1\leq j \leq m$, are disconnected in $G - S$. Hence, $S$ is an edge multiway cut. 
		
		It remains to show that the weight of $S$ is at most $7n+2m$. Since $\mathcal{A}$ satisfies each clause of $\Phi$, there are exactly~$m$ triangle-bases of weight~2 from the clause gadgets in $S$. Similarly, the variable gadgets contribute exactly $n$ bases to $S$. Finally, for each of the $3n$ link-structures, by the definition of $S$
		and the fact that $\mathcal{A}$ is a satisfying assignment,
		either the two connector-edges of the variable-triangle are in $S$ or the two connector-edges of the clause-triangle. Together, they contribute a weight of $6n$ to the total weight of $S$. Therefore, $S$ is an edge multiway cut in $G$ of weight at most $7n+2m$. 
	\end{claimproof}
	
\noindent
Hence, 	$(G, T, k)$ is a {\sc yes} instance of  {\sc  $\{1,2,3,6\}$-Edge Multiway Cut}.
\medskip
	
	Conversely, assume that $(G, T, k)$ is a {\sc yes} instance of  {\sc  $\{1,2,3,6\}$-Edge Multiway Cut}. Hence, there exists an edge multiway cut of $(G,T)$ of weight at most $7n+2m$. We shall demonstrate an assignment that satisfies $\Phi$. Before that, we shall discuss some structural properties of a minimum-weight edge multiway cut. In the following arguments, we assume that the clauses under consideration have size three, unless otherwise specified. While making the same arguments for clauses of size~$2$ is easier, we prefer to argue about clauses of size three for generality.
	
	\begin{claim}[adapted from \cite{DJPSY94}]\label{clm:weight3edges}
		If $e$ is an edge in $G$ incident on a non-terminal vertex $v$ of degree~$> 2$ such that $e$ has weight greater than or equal to the sum of the other edges incident on $v$, then there exists a minimum-weight edge multiway cut in $G$ that does not contain $e$.
	\end{claim} 

	The above claim implies that there exists a minimum-weight multiway cut containing no such edge $e$. To see this, note that an iterative application of the local replacement used in Claim~\ref{clm:weight3edges} would cause a conflict in the event that the replacement is cyclical. Suppose that the edges are replaced in the sequence $e \rightarrow e_1 \rightarrow \ldots \rightarrow e_r \rightarrow e$. Then the weight of $e_1$, denoted by $w(e_1)$ must be strictly less than the weight of $e$. Similarly, $w(e_i)<w(e_j)$ for $i<j$. This would mean that $w(e)<w(e)$, which is a contradiction. 

	\begin{claim}[\cite{DJPSY94}]\label{clm:cycle2edges}
		If a minimum-weight edge multiway cut contains an edge of a cycle, then it contains at least two edges from that cycle.
	\end{claim}

	It follows from Claim~\ref{clm:weight3edges} and the construction of $G$ that there exists a minimum-weight edge multiway cut for $(G,T)$ that does not contain the edges incident on the terminals. Among the minimum-weight edge multiway cuts that satisfy Claim~\ref{clm:weight3edges}, we shall select one that contains the maximum number of connector-edges and from the ones that satisfy both the aforementioned properties, we shall pick one that contains the maximum number of triangle-bases from clause gadgets of size~$2$. Let $S$ be a minimum edge multiway cut that fulfils all these requirements.

	We say a link $u$ incident on a gadget \emph{reaches} a terminal $t$ if $u$ is the first vertex on a path $P$ from the gadget to $t$ and no edge on $P$ is contained in $S$.
	
	A terminal $t$ is \emph{reachable} by a gadget if one of the links incident on the gadget reaches $t$. Note that, for any terminal $t'$ in the gadget, if $t$ is reached from some incident link by a path $P$, then $P$ can be extended to a $t'$-$t$ path in $G$ using only edges inside the gadget. However, among the edges used by such an extension, at least one must belong to $S$, or else $t=t'$.

	\begin{figure}[tbp]
		\centering
		\includegraphics[width=0.55\linewidth]{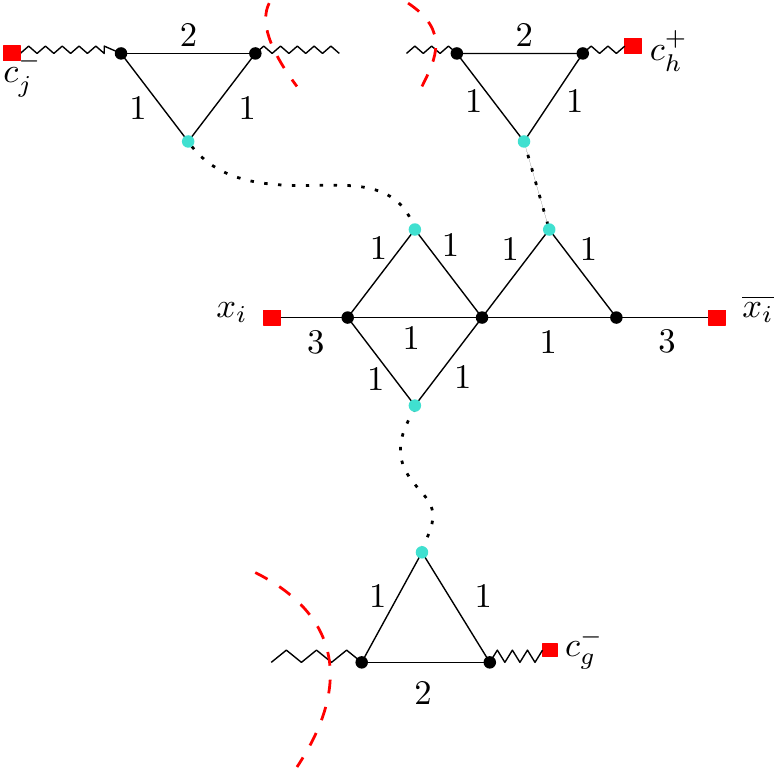}
		\caption{The variable interface of $x_i$. The positive literal $x_i$ occurs in the clauses $c_j$ and $c_g$, whereas $\overline{x_i}$ occurs in $c_h$. The dotted curves connect the two vertices that are identified. No terminal is reachable from the vertex closest to the red dashed lines in the direction of the paths crossed by it.}
		\label{fig:connection}	
	\end{figure}		
	
	\begin{claim}\label{clm:onebasepervar}
		$S$ contains exactly one base of a triangle from each variable gadget.
	\end{claim}
	\begin{claimproof}
	
		Clearly, $S$ must contain at least one base from each variable gadget, else by the fact that $S$ contains no edges incident on terminals, a path between the terminals in such a gadget would remain in $G-S$.
	
		Suppose that $S$ contains two bases of some variable gadget, say that of $x_i$.  By Claim~\ref{clm:cycle2edges}, $S$ must also contain at least three connector-edges from this variable gadget: at least two connector-edges (of the two triangles) of the diamond and at least one connector-edge of the hat. We claim that, without loss of generality, at least all the outer connector-edges must be in $S$. If for some triangle the outer connector-edge next to terminal $t$ is not in $S$, then the link incident on this triangle does not reach any terminal $t' \not= t$; otherwise, a $t$-$t'$ path would remain in $G-S$, a contradiction. Hence, we simultaneously replace all inner connector-edges for which the corresponding outer connector-edge is not in $S$ by their corresponding outer connector-edge. For the resulting set $S'$, the variable terminals of the gadget and their neighbours in $G$ form a connected component of $G-S'$. Since the link incident on a triangle for which the outer connector-edge (next to terminal $t$) was not in $S$ does not reach any terminal $t' \not= t$, $S'$ is feasible. Moreover, it has the same properties we demanded of $S$. Thus, henceforth, we may assume that all the outer connector-edges of the $x_i$-gadget are in $S$. 
		
		We now distinguish six cases based on how many links of the gadget reach a terminal:
			
		\medskip
		\noindent\textsf{Case 1.} {\em No link of the $x_i$ gadget reaches a terminal.} \\
		We can remove one of the two bases from $S$ without connecting any terminal pairs. This is so because in order to disconnect $x_i$ from $\overline{x_i}$, it suffices for $S$ to contain either the base of the diamond along with the two outer connector-edges or the base and outer connector-edge of the hat. No other terminal pairs are connected via the gadget by the assumption of this case. Hence, we contradict the minimality of $S$ (refer to Figure~\ref{fig:Case1}).

		\begin{figure}[t]
			\centering
			\includegraphics[width=\textwidth]{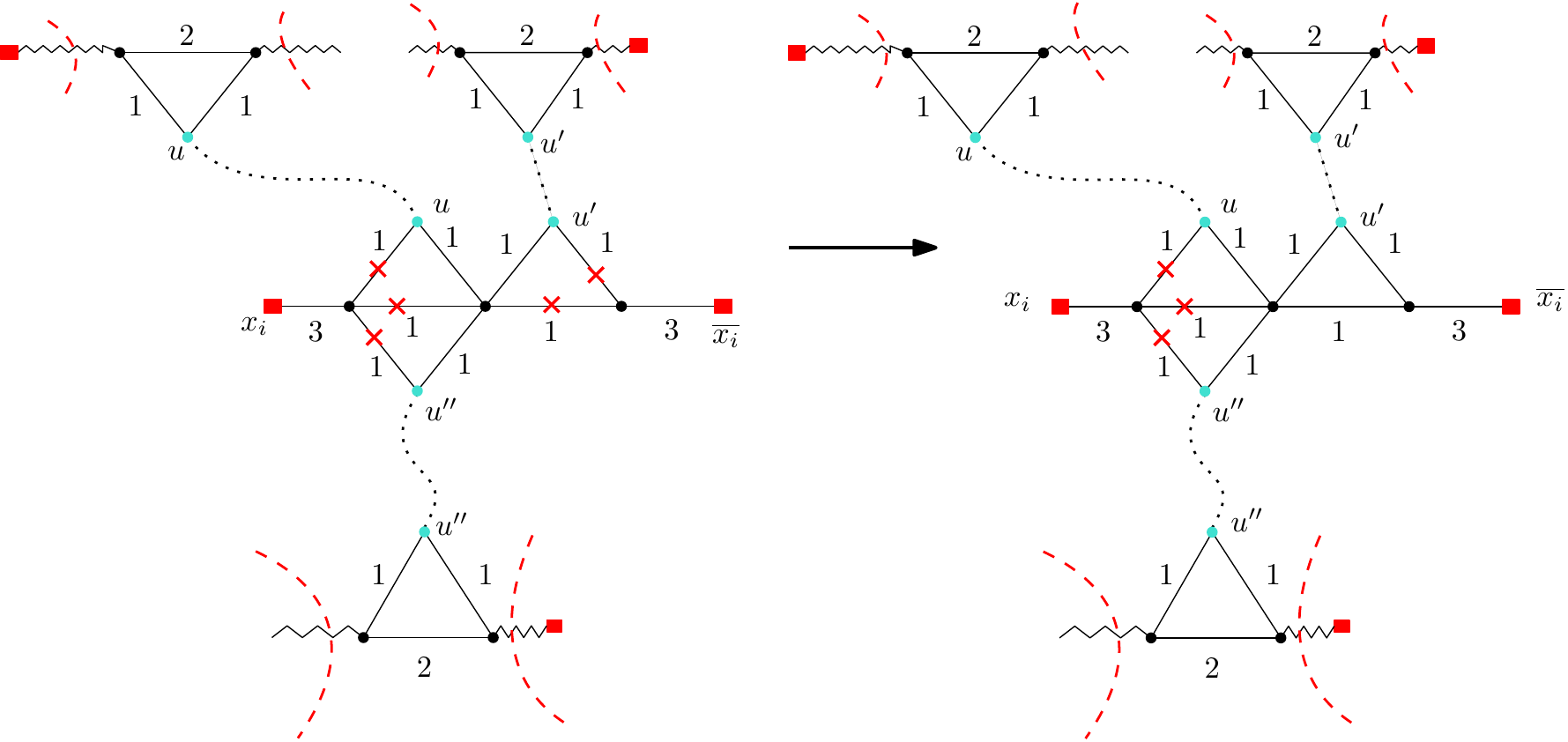}
			\caption{Case 1: In the figure on the left, we see the $x_i$-gadget with the three clause gadgets it is linked to. The dotted lines indicate that the links are identified with each other. None of the links can reach any terminal. The red dashed curves indicate that the path is intersected by the multiway cut $S$. The edges labeled with a red cross are contained in $S$. In the right figure, we show how $S$ can be modified without compromising its feasibility.}
			\label{fig:Case1}
		\end{figure}

		\begin{figure}
			\centering
			\includegraphics{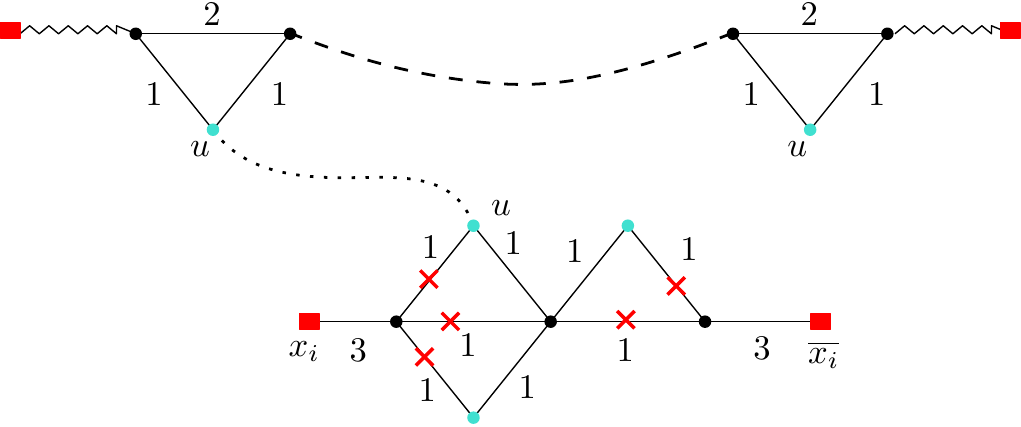}
			\caption{Case 2: In the figure we see the $x_i$-gadget with one of its links reaching two distinct terminals. The dotted curve indicates that the links are identified with each other. The dashed curve shows that there exists a path between its endpoints.}
			\label{fig:Case2}
		\end{figure}
		
		\medskip\noindent \textsf{Case 2.} {\em A link of the $x_i$-gadget reaches at least two distinct terminals.} \\
		By the definition of reaches, this implies that there is a path in $G-S$ between any two of the reached terminals (see Figure~\ref{fig:Case2}). This contradicts that $S$ is an edge multiway cut for $(G,T)$.
		
		\begin{figure}[t]
			\centering
			\includegraphics[width=\textwidth]{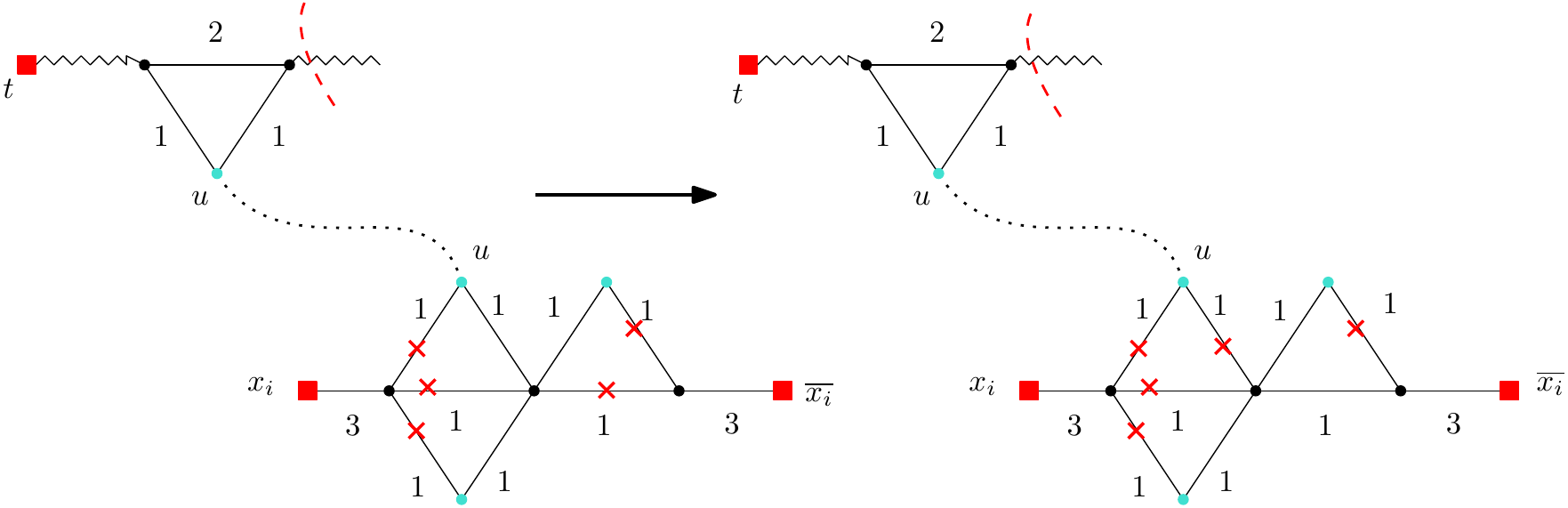}
			\caption{Case 3: The left figure shows the situation when exactly one link of the $x_i$-gadget reaches a terminal. The edges labeled with a red cross are contained in $S$. The right figure shows the replacement made in this case.}
			\label{fig:case3}
		\end{figure}

		\medskip\noindent \textsf{Case 3.}{\em Exactly one link $u$ of the $x_i$-gadget reaches some terminal $t$.} \\
		We remove from $S$ the base of a triangle that is not attached to $u$ and add the remaining connector-edge of the triangle that is attached to $u$ (if it is not already in $S$). Refer to Figure~\ref{fig:case3}. Consequently, although $u$ reaches $t$, both connector-edges incident on $u$ are in $S$. Since no other link reached any terminals and $x_i$ remains disconnected from $\overline{x_i}$ in $G-S$, we can obtain an edge multiway cut for $(G,T)$ satisfying Claim~\ref{clm:weight3edges} that has the same or less weight as $S$, but has strictly more connector-edges than $S$. This is a contradiction to our choice of $S$.
		
		\begin{figure}[t]
			\centering
			\includegraphics[width=\textwidth]{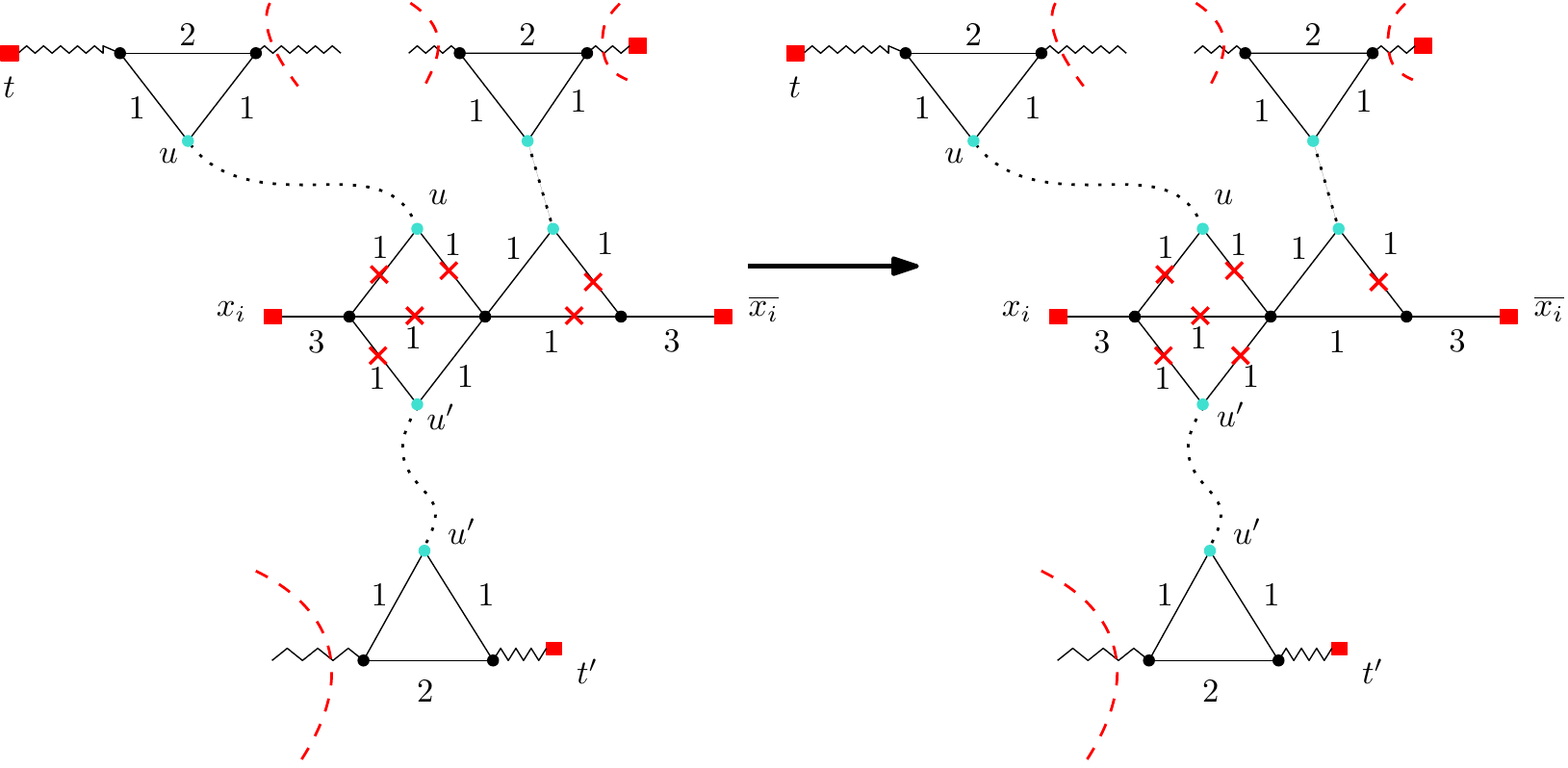}
			\caption{Case 4: The left figure shows the situation when exactly two links of the $x_i$-gadget reach two distinct terminals. The edges labeled with a red cross are contained in $S$. The right figure depicts the situation after the replacement.}
			\label{fig:case4}
		\end{figure}
		\medskip
		\noindent \textsf{Case 4.} {\em Exactly two links $u, u'$ of the $x_i$-gadget reach two distinct terminals $t$ and $t'$, respectively.} \\
		Recall that all three outer connector-edges are in $S$. Now at least one of the inner connector-edges of the gadget must be in $S$, or else $t$ would be connected to $t'$ via this gadget. In particular, both the connector-edges of at least one of the two triangles attached to $u, u'$ must be in $S$. Figure~\ref{fig:case4} depicts this scenario. We can remove from $S$ one of the two bases and add instead the remaining connector-edge of the other triangle (if it is not already in $S$). Consequently, although $u$ reaches $t$ and $u'$ reaches $t'$, all connector-edges incident on $u$ and $u'$ are in $S$. Moreover, $x_i$ and $\overline{x_i}$ are not connected to each other in $G-S$, as one base and its corresponding outer connector(s) are still in $S$. The transformation results in an edge multiway cut for $(G,T)$ satisfying Claim~\ref{clm:weight3edges} that has the same or less weight than $S$, but has strictly more connector-edges than $S$. This is a contradiction to our choice of $S$.
		
		\begin{figure}[t]
			\centering
			\includegraphics[width=\textwidth]{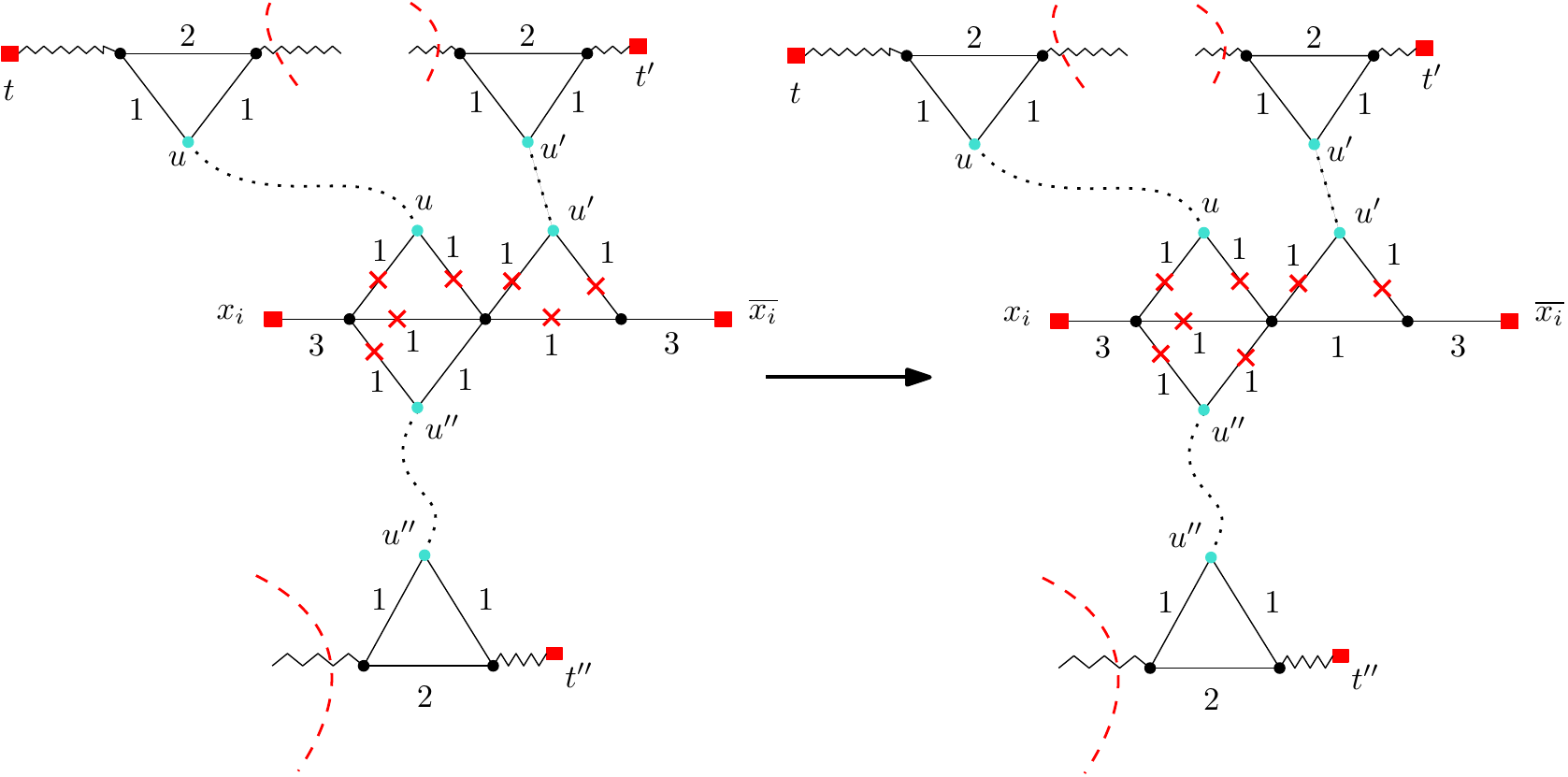}
			\caption{Case 5: The left figure shows the situation when all the three links of the $x_i$-gadget reach three distinct terminals. The edges labeled with a red cross are contained in $S$. The right figure shows the situation after the replacement.}
			\label{fig:case5}
		\end{figure}

		\medskip
		\noindent \textsf{Case 5.} {\em All the three links of the $x_i$-gadget reach distinct terminals $t,t',t'' $, respectively.}\\ 
		Recall that all three outer connected edges are in $S$. Now at most one (inner) connector-edge of the $x_i$-gadget is not in $S$, or else at least one pair of terminals among $\{(t,t'), (t',t''), (t'' ,t)\}$ would remain connected via the gadget. Consider Figure~\ref{fig:case5} for a visual depiction of this case. We replace one of the bases in $S$ with this connector-edge (if it is not already in $S$). The resulting edge multiway cut is no heavier. To see that it is also feasible, note that while $t, t', t''$ are still reached from the links of the gadget, all the connector-edges of this gadget are in the edge multiway cut. The terminals $x_i$ and $\overline{x_i}$ are disconnected from each other in $G-S'$ because one triangle-base and its connectors are still in the edge multiway cut. Hence, we obtain an edge multiway cut for $(G,T)$ satisfying Claim~\ref{clm:weight3edges} that has the same or less weight than $S$, but with strictly more connector-edges than $S$, a contradiction to our choice of $S$.
		
		\begin{figure}[t]
			\centering
			\includegraphics[width=\textwidth]{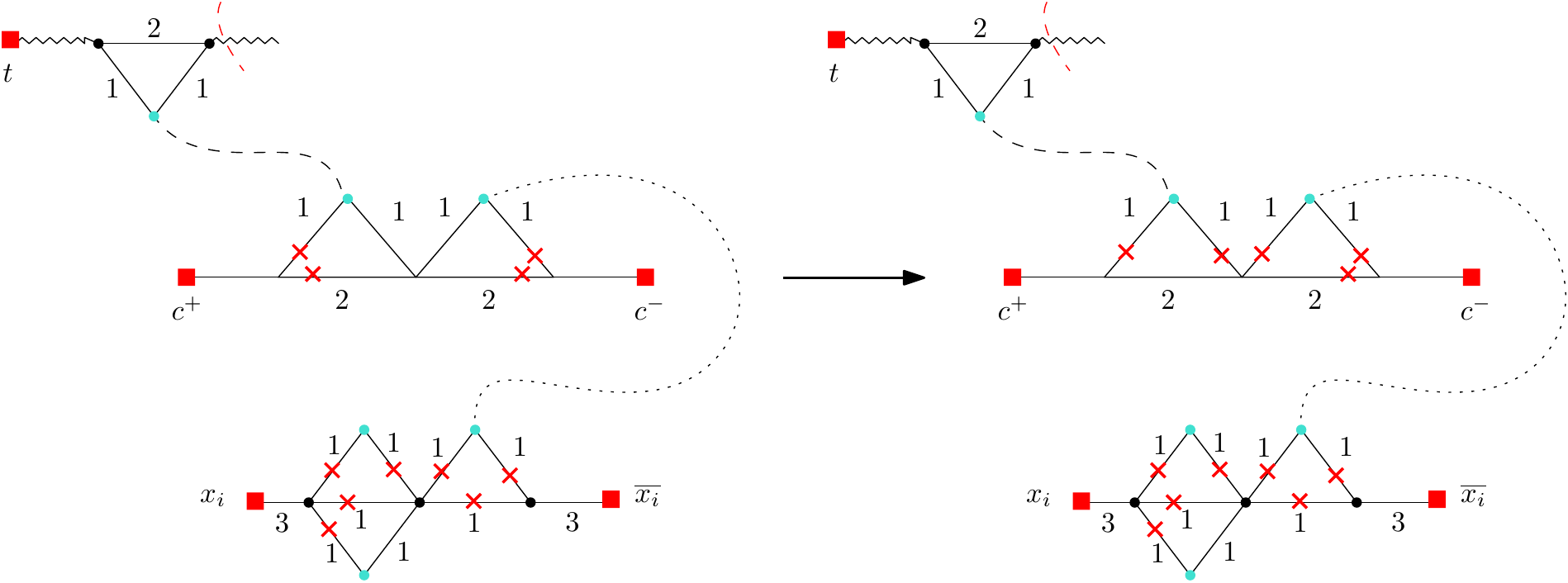}
			\caption{Case 6: The figure on the left shows the situation when the $x_i$-gadget reaches a terminal $t$ via a clause gadget of size two. The dotted curve in the figure indicates that its endpoints are identified whereas the dashed curve indicated that there exists a path between its endpoints that is not cut by $S$. The figure on the right depicts the situation after the replacement.}
			\label{fig:case6-1}
		\end{figure}
		\medskip
		\noindent \textsf{Case 6.} {\em At least two links of the $x_i$-gadget reach exactly one terminal $t$ outside the gadget.} \\
		Recall that every variable occurs in at least two clauses of size~$2$. Hence, $t$ is reachable via a link from the $x_i$-gadget to at least one directly linked clause gadget of a clause of size~$2$. Also recall that $S$ is a minimum-weight edge multiway cut containing the maximum number of bases from clauses of size~$2$. 
		
		Suppose that there exists a size-$2$ clause gadget $c$, directly linked to the $x_i$-gadget, that does not contain $t$ and via which $t$ is reachable from the $x_i$-gadget. 
		That is, some link reaches $t$ via a path $P$ that contains edges of $c$, but $t$ is not in $c$. Refer to Figure~\ref{fig:case6-1} for a visual depiction.
		Then $S$ must contain two base-connector pairs from $c$; else, some terminal of $c$ would not be disconnected from $t$ in $G - S$. Now remove from $S$ the base of one of the two triangles of $c$ and add the remaining two connector-edges of $c$. This does not increase the weight, as the base of the clause-triangle has weight~$2$ and the connectors have weight~$1$ each. The only terminal pair that could get connected by the transformation is the pair of terminals on $c$ itself. However, one of the bases is still in the transformed cut. This new cut contradicts our choice of $S$, as it has strictly more connector-edges and satisfies the other conditions.
		
		Suppose $t$ is contained in one of the size-$2$ clause gadgets $c'$, directly linked to the $x_i$-gadget. If the link between the $x_i$-gadget and $c'$ is not one of the links meant in the assumption of this case, then the situation of the previous paragraph holds and we obtain a contradiction.
		Thus, $t$ is reachable from the $x_i$-gadget via both links of $c'$. 
		Hence, a base-connector pair of the triangle of $c'$ that $t$ is not attached to must be in $S$. Consider the link of the $x_i$-gadget that is not linked to $c'$ but reaches $t$ and let $P$ be a corresponding path, starting at this link and going to $t$. Note that $P$ passes through a clause gadget $c''$ directly linked to the $x_i$-gadget. If $c''$ is a size-$2$ clause gadget, then we obtain a contradiction as before. Hence, $c''$ corresponds to a size-$3$ clause (as in Figure~\ref{fig:mwc-claimproof}). Since $P$ must either enter or leave $c''$ through one of its outer triangles, a base-connector pair of at least one outer triangle of $c''$ must be in $S$, or the attached terminal would reach $t$ in $G-S$, contradicting that $S$ is an edge multiway cut for $(G,T)$. Let $\Lambda$ be such an outer triangle (see Figure~\ref{fig:mwc-claimproof}).			

		\begin{figure}[tbp]
			\centering
			\includegraphics[width=0.65\linewidth]{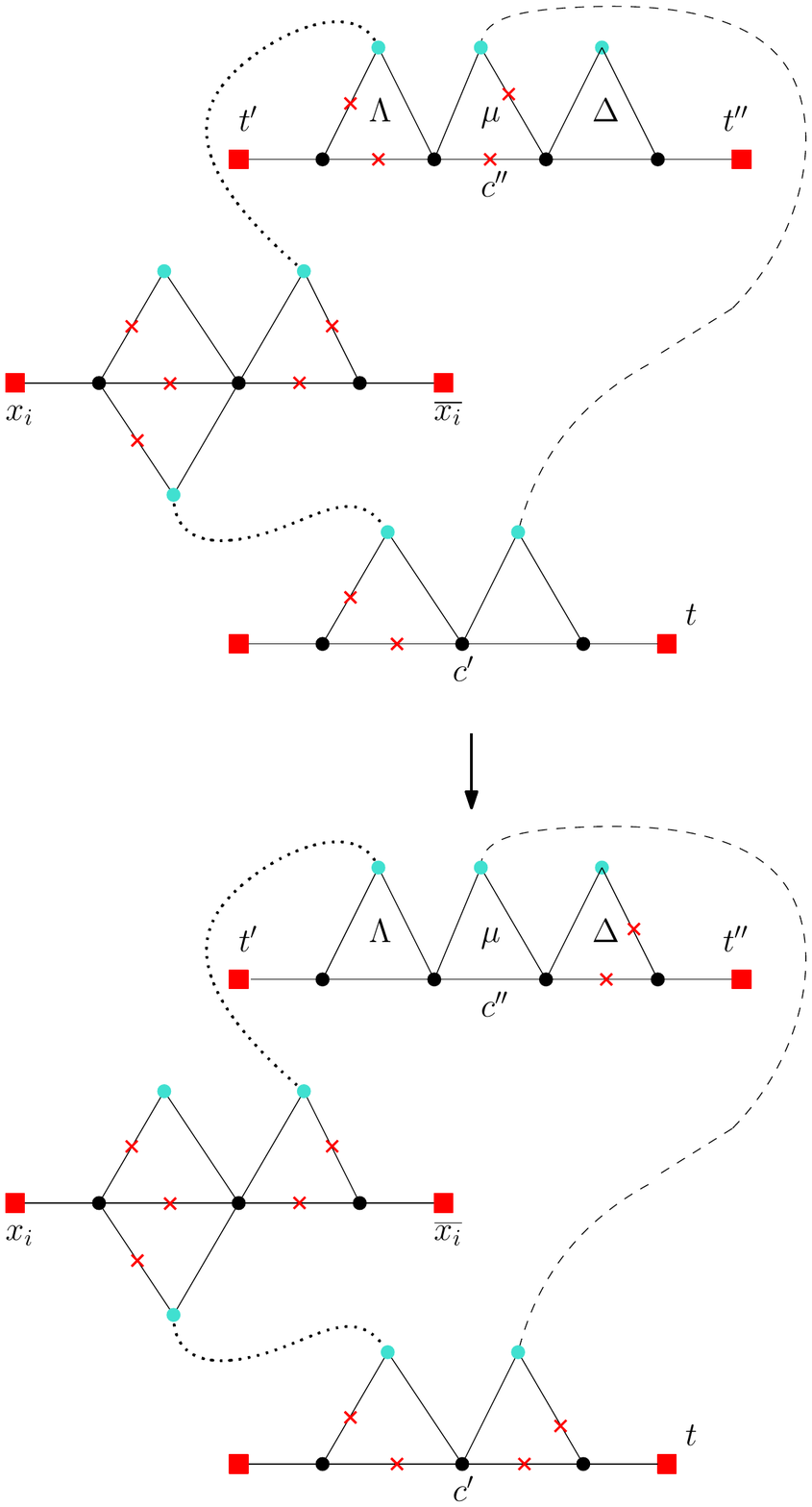}
			\caption{Case 6: In the top figure, there is a terminal $t$ reachable via (at least) two links of the $x_i$-gadget. Moreover, $t$ appears in a clause gadget $c'$ corresponding to a clause of size two that is directly connected to the $x_i$-gadget. The endpoints of the dotted curves are identified. The dashed curve indicates the existence of a path, not cut by $S$, between its endpoints.}
			\label{fig:mwc-claimproof}	
		\end{figure}	

		We argue that, without loss of generality, $S$ contains a base-connector pair of the other outer triangle, $\Delta$. Suppose not. Then, in particular, the base of $\Delta$ is not in $S$. If $P$ passes through the link attached to $\Delta$, then one of the endpoints of the base of $\Delta$ must be on $P$. Since the base of $\Delta$ is not in $S$, the terminal $t''$ next to $\Delta$ remains connected to $t$ in $G - S$, a contradiction. Hence, $P$ must either enter or exit $c''$ via the link attached to its middle triangle $\mu$. Moreover, $S$ must contain a base-connector pair of $\mu$ (see Figure~\ref{fig:mwc-claimproof}), or $t''$ would still reach $t$ in $G-S$. We now modify $S$ to obtain a set $S'$. If both connector-edges of $\Delta$ are in $S$, then replace the base of $\mu$ by the base of $\Delta$ to obtain $S'$. Then all edges of $\Delta$ are in $S'$. Otherwise, no edge of $\Delta$ is in $S$ and thus no terminal must be reachable via the link attached to $\Delta$ (or it would be connected to $t''$ in $G-S$). So, we replace the base-connector pair of $\mu$ by a base-connector pair of $\Delta$ to obtain $S'$. Then $S'$ is an edge multiway cut for $(G,T)$ of the same weight at $S$ that has the same properties as $S$. Hence, we may assume $S=S'$. Then $S$ contains a base-connector pair of $\Delta$.

		Now remove from $S$ the base and connector-edge of $\Lambda$. Then $t$ and $t'$ become connected to each other in $G-S$, but not to any other terminal, or that terminal would already be connected to $t$ in $G-S$. Now add the base and outer connector-edge of the triangle in $c'$ that $t$ is attached to. This restores that $S$ is an edge multiway cut for $(G,T)$. 
The edge multiway cut we obtain has the same weight as $S$ and satisfies Claim~\ref{clm:weight3edges}. Moreover, it has no less connectors than $S$ but contains at least one more base of a clause gadget of size~$2$. Hence, we obtain a contradiction to our choice of~$S$.
\end{claimproof}
	
	We now focus on the link-structures.

	\begin{claim}\label{clm:middletriangle}
		There cannot exist a link-structure in $G$ that contributes less than two edges to $S$ and for which the clause-triangle of the link-structure contributes no connector-edges to~$S$.
	\end{claim}
	\begin{claimproof}
		Towards a contradiction, suppose that such a link-structure exists. Let the clause gadget containing the link-structure be $c$ and the variable gadget containing it be $x_i$. By Claim~\ref{clm:onebasepervar}, we know that there exists a triangle of the $x_i$-gadget that does not contribute its base to $S$. Therefore, at least one terminal $t$ of the $x_i$-gadget is reachable from the clause gadget $c$. This implies that the clause-triangle of the link-structure is the middle triangle of $c$; 
		else, there would exist a path in $G-S$ between $t$ and the closest clause-terminal on $c$, because the edge incident on this terminal is also not in $S$ by its properties. 
		Then, since $S$ is feasible, it must contain the base and at least one connector-edge of each of the two outer triangles of $c$. Else, at least one of the clause-terminals would be reachable from $t$ in $G-S$. 
		
		It must also be the case that both connector-edges of each of the outer triangles must be in $S$ or the incident link reaches no terminal $t' \not= t$; otherwise, $t$ or the incident clause-terminal would be connected to $t'$ in $G-S$. 
		Now, we can remove one of the two bases from $S$ and add the two connector-edges of the middle triangle, without compromising the feasibility of the edge multiway cut. Thus, there exists an edge multiway cut of no greater weight than $S$, satisfying Claim~\ref{clm:weight3edges}, and containing two more connector-edges (those of the clause-triangle of the link-structure). This is a contradiction to our choice of $S$.
	\end{claimproof}

	\begin{figure}
		\centering
		
		\includegraphics[width=0.35\textwidth]{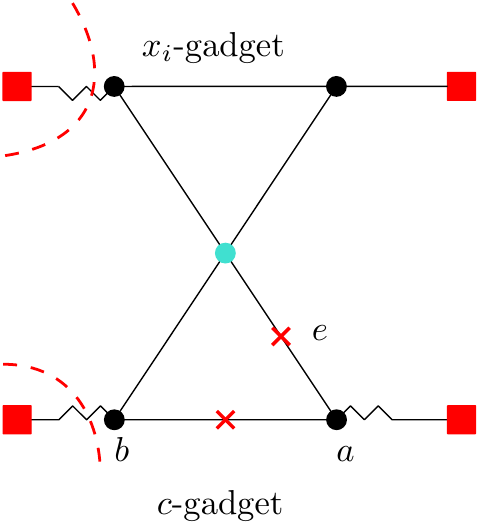}
		\caption{The figure depicts a link-structure with the variable gadget of $x_i$ at the top and its clause gadget for $c$ at the bottom. Exactly one edge of the link-structure (labeled with a red cross) is in the set $S$. The dashed red lines depict that the terminals cannot be reached from the vertices $a$ or $b$.}\label{fig:2perlink}
	\end{figure}

	\begin{claim}\label{clm:at least 2 per link}
		$S$ contains at least two edges from each link-structure.
	\end{claim}
	\begin{claimproof}
	Suppose that there exists a link-structure $\ell$ that contributes less than two edges to $S$. Suppose that $\ell$ connects the clause gadget $c$ and the variable gadget $x_i$. By Claim~\ref{clm:middletriangle}, we know that the clause-triangle of $\ell$ must contribute an edge $e$ to $S$. Therefore, none of the connectors of the variable-triangle attached to $\ell$ are in $S$. As a result, the variable-terminal of the $x_i$-gadget attached to $\ell$, say we call it $t$, is reachable from $c$ via $\ell$.
		
	By Claim~\ref{clm:cycle2edges} and the fact that only $e$ is in $S$, the base of the clause-triangle must also be in $S$. We do the following replacement: remove from $S$ the base-connector pair of the clause-triangle and add the base and (possibly two) connectors of the variable-triangle of $\ell$, as follows. If the variable-triangle of $\ell$ is part of a diamond, then we add to $S$ the base and two outer connectors, thereby getting an edge multiway cut of equal weight but strictly more connectors. If the variable-triangle is a hat, then we add to $S$ the base and outer connector of the hat, obtaining an edge multiway cut for $(G,T)$ of strictly smaller weight than $S$. If we can show that the resultant edge multiway cut is feasible, we obtain a contradiction in either scenario. We claim that such a replacement does not compromise the feasibility of $S$.
		
	Let $a, b$ be the endpoints of the base of the clause-triangle of $\ell$, where $a$ is the endpoint on which $e$ is incident (see Figure~\ref{fig:2perlink}).
	Note that no terminal other than $t$ should be reachable in $G - S$ from $b$; else, there would be a path from $t$ to that terminal via $\ell$. In particular, the terminal of the clause gadget for $c$ on the side of $b$ cannot be reached in $G-S$ from the vertex $b$. By removing the base-connector pair of the clause-triangle of $\ell$, we may expose the clause-terminal on the side of the vertex $a$ (or another terminal outside $c$) to $t$. However, by adding the base and (possibly two) connectors closest to $t$, we disconnect any path between this terminal and $t$. Since we did not modify the cut in any other way, no new connections would have been made. This shows the feasibility of the resultant edge multiway cut and thus proves our claim. 		
	\end{claimproof}			
	
	\begin{claim}\label{clm:final}
		If there exists an edge multiway cut of weight at most $7n+2m$ for $(G,T)$, then there exists a satisfying truth assignment for $\Phi$.
	\end{claim}
	\begin{claimproof}
		Let $S$ be the edge multiway cut defined before. The immediate consequence of Claims~\ref{clm:onebasepervar} and~\ref{clm:at least 2 per link} is that the weight of $S$ is at least $n + 2\cdot(3n)= 7n$. $S$ must also contain at least one base per clause gadget lest the two terminals on a clause gadget remain connected. Therefore, its weight is at least $7n+2m$. Since it is an edge multiway cut of weight at most $7n+2m$, it has exactly one base per clause gadget. 
		
		We also claim that for each link-structure, if one of the triangles attached to it has its base in $S$, then the other one cannot: note that if both the triangles had their bases in $S$, then each of them would also have a connector-edge in $S$ by Claim~\ref{clm:cycle2edges}. By Claim~\ref{clm:at least 2 per link} and the assumption that the weight of $S$ is at most $7n+2m$, the other two connector-edges of the link-structure are not in $S$. Since at most one base per variable/clause gadget can be in $S$, there would be a path between one of the variable-terminals and one of the clause-terminals in the linked gadgets through the link-structure, a contradiction to $S$ being an edge multiway cut for $(G,T)$. Figure~\ref{fig:oneb/ls} shows one such case.

		\begin{figure}
			\centering
			\includegraphics[width=0.55\textwidth]{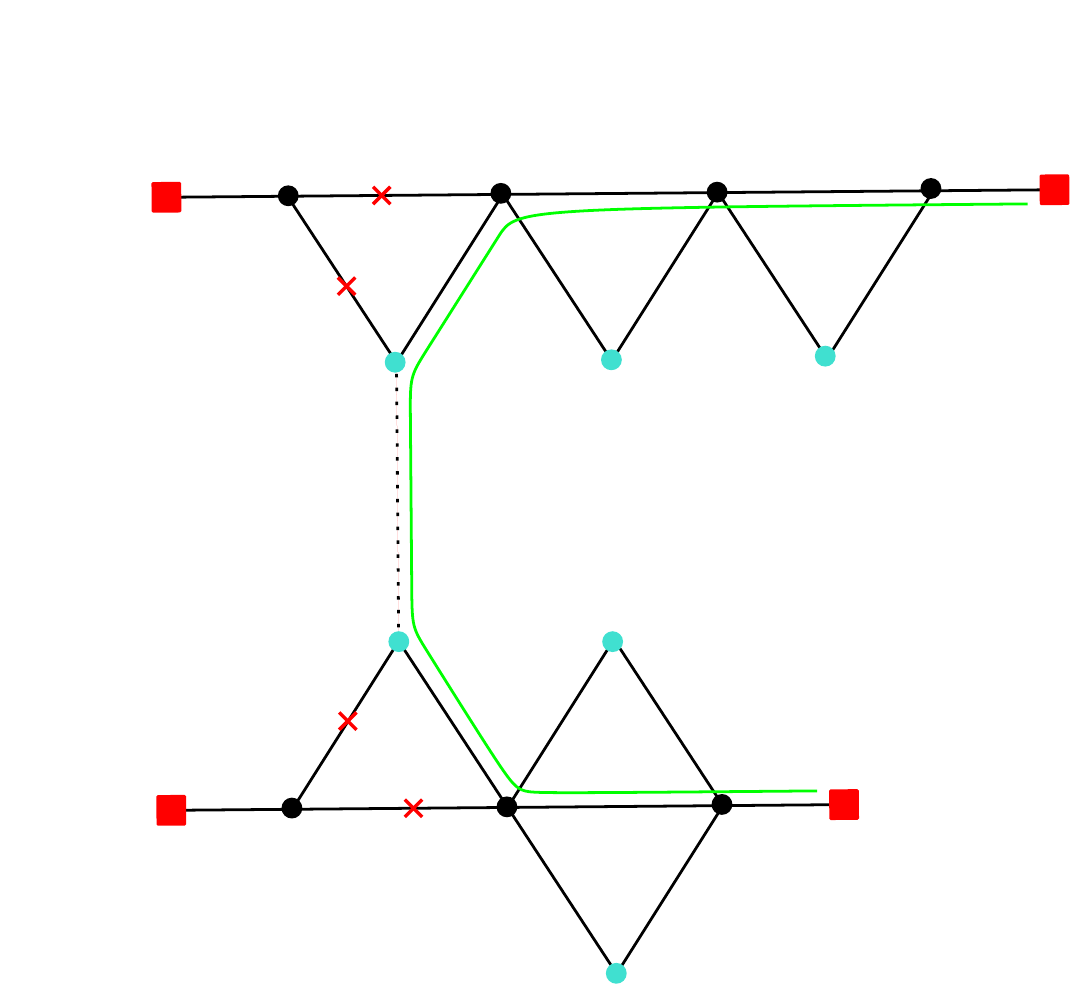}
			\caption{The figure shows a link-structure with the variable gadget at the bottom and its connected clause gadget at the top. The crossed-out red edges are the ones contained in the minimum edge multiway cut $S$. The green curve shows the existence of a path between a variable-terminal and a clause-terminal. The dotted curve connects the identified connectors in the link-structure shown in the figure.}\label{fig:oneb/ls}
		\end{figure}

		We now define the truth assignment $\mathcal{A}$. For each variable-terminal, if the diamond has its base in $S$, we make it ``false'', otherwise if the hat has its base in $S$ we make it ``true". Each clause gadget has exactly one triangle contributing its base to $S$. From the above argument, we know that the variable-triangle linked to this clause-triangle must not contribute its base to $S$. Hence, every clause gadget is attached to one literal triangle such that its base is not in $S$, and is therefore ``true''. Hence, every clause is satisfied by the truth assignment $\mathcal{A}$ and $\Phi$ is a  {\sc yes} instance of {\sc Planar 2P1N-3SAT}.
	\end{claimproof}
	
	The above implies that {\sc $\{1,2, 3, 6\}$-Edge Multiway Cut} is \NP-complete on planar subcubic graphs. We now proceed to prove that (unweighted) {\sc Edge Multiway Cut} is \NP-complete on planar subcubic graphs. The proof follows from the
claim below, which states
that the honeycombs of $\tilde{G}$ (defined before) do not contribute any edge to any minimum edge multiway cut for ($\tilde{G},T$).

	\begin{restatable}{claim}{honeycomb}\label{clm:nohoney}
		Any minimum edge multiway cut for $(\tilde{G},T)$ does not contain any of the honeycomb edges.
	\end{restatable}
	
\begin{claimproof}
	Let $S'$ be a minimum edge multiway cut for $(\tilde{G},T)$. Recall that $\tilde{G}$ is planar. Note that for any two vertices $s,t$, an $s$-$t$ cut in a planar graph corresponds to a simple (possibly degenerate) cycle in the planar dual~\cite{Re83}. Therefore, the dual of an edge multiway cut comprises several cycles. Let the edges corresponding to $S'$ in the planar dual of $\tilde{G}$ be $S^*$. In fact, $S^*$ induces a planar graph such that exactly one terminal of~$T$ is embedded in the interior of each face of this graph. If any face of $S^*$ did not contain a terminal, we could remove the edge in $S'$ dual to one of the edges of this face. This would not connect any terminal pair, and hence contradicts the minimality of $S'$.
	
	Suppose that $S'$ contains some of the edges of the honeycomb in $\tilde{G}$ replacing the vertex $v \in V(G')$. We denote the intersection of  $S'$ with the edges of this honeycomb by $S'_h$. Let the set of edges dual to $S'_h$ in be  $S^*_h$. By abuse of notation, we also denote by $S^*_h$ the graph formed by contracting all the edges in $S^* \setminus S^*_h$. Since each face of $S^*$ encloses a terminal, each bounded face of $S^*_h$ must enclose an attachment point of the honeycomb. If not, then we could remove from $S'$ an edge in $S'_h$ dual to some edge of the face of $S^*_h$ not enclosing an attachment point. This does not make any new terminal-to-terminal connections, as the part of the honeycomb enclosed by this face does not contain any path to any of the terminals of~$T$. This would be a contradiction to the minimality of $S'$.  
	
	Next, we observe that no bounded face of $S^*_h$ can enclose more than one attachment point. Suppose that there exists a bounded face in $S^*_h$ that encloses two attachment points. Since the two attachment points are separated by 100 cells of the honeycomb, the length of the face boundary must be at least 50. We could remove all the 50 edges from $S'$ dual to the edges of the face boundary and add all the attaching edges to $S'$, instead. All the terminal-to-terminal paths passing through the honeycomb will remain disconnected after the transformation. Since at most eight attaching edges can be added, we again get a contradiction to the minimality of $S'$. So, each bounded face of $S^*_h$ must enclose exactly one attachment point. 
	
	To enclose the attachment points, each of these faces must cross the boundary of the honeycomb exactly twice. We claim that the faces of $S^*_h$, enclosing consecutive attachment points on the boundary of the honeycomb, are pairwise edge-disjoint. Suppose that the faces enclosing two consecutive attachment points, $a$ and $a'$, share an edge. Then, they must also share an edge that crosses the boundary of the honeycomb. If they do not, then let $e$ be the last edge of the face enclosing $a$ to cross the boundary and $u'$ be the first edge of the face enclosing $a'$ to cross the boundary of the honeycomb. The edges $e$ and $e'$ along with the other edges not shared between the respective face boundaries bound a region of the plane containing no attachment points, a contradiction! 
	
	Therefore, any two faces of $S^*_h$ enclosing consecutive attachment points share an edge which crosses the boundary of the honeycomb. Without loss of generality, let this edge be closer to $a$. Then, the face enclosing $a'$ must contain at least 50 edges as $a$ and $a'$ are separated by 100 cells of the honeycomb. This implies that $S'_h$ contains at least 50 edges. However, we could remove from it all the 50 edges and add all the (at most eight) attaching edges. This cut is smaller in size and disconnects all the terminal-terminal paths passing through the honeycomb. Once again, we contradict the minimality of $S'$.
	
	Hence, all the faces in $S^*_h$ enclosing attachment points are edge-disjoint. So, there are at least $2\cdot \operatorname{deg}_{G'}(v)$ edges in $S'_h$. We could replace this cut by a smaller cut, namely, the edge multiway cut formed by removing the edges in $S'_h$ from $S'$ and adding to it all the attaching edges incident on the attachment points. This cut disconnects all terminal-paths passing through the honeycomb and yet, is smaller in size than $S'$, a contradiction to its minimality. Hence, $S'$ does not contain any edge of any of the honeycombs.
\end{claimproof}	
	
	By the construction of $\tilde{G}$ and Claims~\ref{clm:forward}, \ref{clm:final}, and~\ref{clm:nohoney}, we conclude that {\sc Edge Multiway Cut} is \NP-complete on planar subcubic graphs.
\end{proof}

\section{The Proof of Theorem~\ref{thm:NMwC:C2}}\label{s-2}

In this section we prove Theorem~\ref{thm:NMwC:C2}. We start with the following observation.

\begin{restatable}{proposition}{NMWCprop}\label{prp:NMWC4}
	{\sc Node Multiway Cut} is \NP-complete for planar graphs of maximum degree~$4$.
\end{restatable}

\begin{proof}
	It is readily seen that {\sc Node Multiway Cut} belongs to \NP. We now reduce from {\sc Node Multiway Cut with Deletable Terminals} on planar subcubic graphs. Let $(G,T,k)$ be an instance of this problem. Let $G'$ be obtained from $G$ by adding a pendant vertex $v'$ per vertex $v \in T$. Let $T' = \{v' \mid v \in T\}$. If $(G',T')$ has a node multiway cut $S \subseteq V(G') \setminus T'$, then $S$ is immediately a node multiway cut for $(G,T)$. Conversely, if $(G,T)$ has a node multiway cut $S \subseteq V(G)$, then $S$ is immediately a node multiway cut for $(G',T')$ with $S \subseteq V(G') \setminus T'$. The result follows.
\end{proof}

\noindent
We also need the following lemma from Johnson et al.~\cite{JMOPPSV} 
(the proof of this lemma can also be found in the appendix).

 \begin{restatable}{lemma}{EMWClem}\label{lem:Mwc:C3}
	If {\sc Edge Multiway Cut} is \NP-complete for a class $\mathcal{H}$ of graphs, then it is also \NP-complete for the class of graphs consisting of the $1$-subdivisions of the graphs of~$\mathcal{H}$.
\end{restatable}

\begin{figure}[t]
	\centering
	\includegraphics[width=\textwidth]{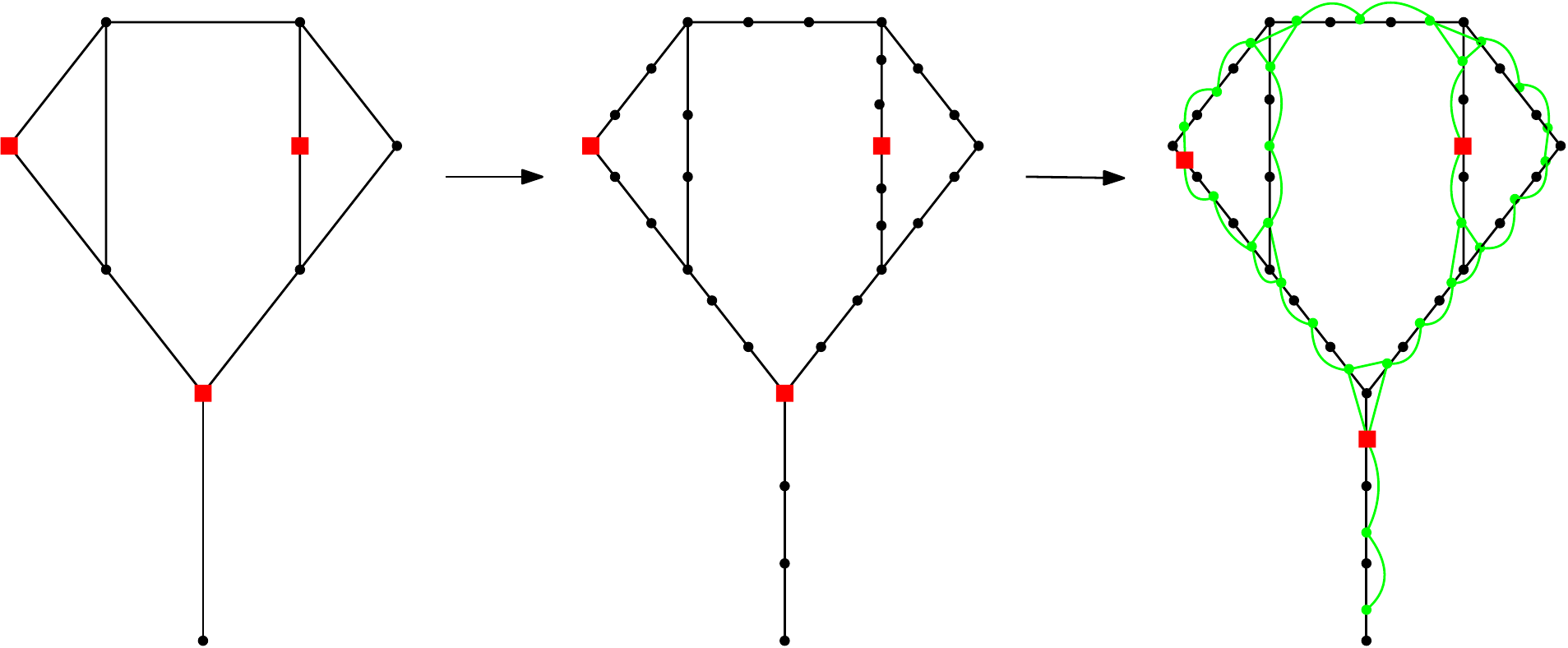}
	\caption{The figure shows the construction in Theorem~\ref{thm:NMwC:C2}. The leftmost figure is an instance of {\sc Edge Multiway Cut} on planar subcubic graphs. The figure in between shows a $2$-subdivision of the instance. The rightmost figure shows the line graph of the subdivided graphs drawn in green. In each figure, the terminals are shown as red squares.}
	\label{fig:nmwchard}
\end{figure}
\noindent
We are now ready to prove Theorem~\ref{thm:NMwC:C2}.

\NMWCthm*

\begin{proof}
	It is readily seen that {\sc Node Multiway Cut} belongs to \NP.
	In Theorem~\ref{thm:MwC:C2}, we showed that {\sc Edge Multiway Cut} is \NP-complete on the class of planar subcubic graphs. We will now reduce {\sc Node Multiway Cut} from {\sc Edge Multiway Cut} restricted to the class of planar subcubic graphs. Let $G$ be a planar subcubic graph with a set of terminals~$T$.  
	
From $G$, we create an instance of {\sc Node Multiway Cut} by the following operations; 
here, the {\it line graph} of a graph $G=(V,E)$ has $E$ as vertex set and for every pair of edges $e$ and $f$ in $G$, there is an edge between $e$ and $f$ in the line graph of $G$ if and only if $e$ and $f$ share an end-vertex.
	
	\begin{itemize}
		\item 
		We construct the $2$-subdivision of $G$, which we denote by	$G'$.
		\item Next, we construct the line graph of $G'$, which we denote by $L$. 
		\item Finally, we create the terminal set of $L$ as follows: for each terminal $t$ in $G'$, consider the edges incident on it. In the line graph $L$, these edges must form a clique, $K_i$ for $i \in \{1,2,3\}: i= \operatorname{deg}(t)$. In this clique, we pick one vertex and make it a terminal. We denote the terminal set in $L$ by $T_L$. 
	\end{itemize}

\noindent
Note that $L$ is planar, as $G'$ is planar and every vertex in $G'$ has degree at most~$3$~\cite{S64}.
Note also that $L$ is subcubic, as every edge in $G'$ has one end-vertex of degree~$2$ and the other end-vertex of degree at most $3$.
Moreover, $L$ and $T_L$ can be constructed in polynomial time. 

	\begin{restatable}{claim}{NMWCclm}\label{clm:nmwchard}
		There exists an edge multiway cut of $(G, T)$ of size at most $k$ if and only if there exists a node multiway cut of $(L, T_L)$ of size at most $k$.
	\end{restatable}

\begin{claimproof}
	We assume that $(G,T)$ has an edge multiway cut $S$ of size at most $k$. By Lemma~\ref{lem:Mwc:C3}, $G'$ also has an edge multiway cut of size at most $k$. We claim that there exists an edge multiway cut $S'$ of $G'$ of size at most $k$ which does not contain any edge incident on a terminal. Every edge in $G'$ is adjacent to some edge with both its ends having degree two. Therefore, if an edge in the edge multiway cut of $G'$ is incident on a terminal, we can replace it with its adjacent edge, which disconnects all the paths disconnected by the former and does not increase the size of the edge multiway cut.  Now, for each edge in $S'$ we add its corresponding vertex in $L$ to a set $S_L$. Since $S'$ pairwise disconnects the terminals in $G'$, $S_L$ disconnects all the terminal cliques from each other. Therefore, $S_L$ is a node multiway cut of~$L$.
	
	Conversely, let $S'_L \subseteq V(L) \setminus T_L$ be a node multiway cut of $(L,T_L)$ of size at most $k$. By similar arguments as above, we may assume that $S'_L$ does not contain any vertex from any terminal-clique. We claim that $G$ has an edge multiway cut of size at most $k$. To that end, we show that $G'$ has an edge multiway cut of size at most $k$ and apply Lemma~\ref{lem:Mwc:C3} to prove the same for $G$. We add to the edge multiway cut $S$ the edges of $G'$ that correspond to the vertices in $S'_L$. The size of $S$ is clearly at most $k$. To see that it is an edge multiway cut of $G'$, note that pairwise disconnecting the terminal-cliques of $L$ amounts to pairwise disconnecting the set of edges incident on any terminal in $G'$ from its counterparts. This, in turn, pairwise disconnects all the terminals in $G'$.
\end{claimproof}	
	
\noindent
By our construction and Claim~\ref{clm:nmwchard}, {\sc Node Multiway Cut} is \NP-complete on the class of planar subcubic graphs.
\end{proof}

\section{The Proof of Theorem~\ref{thm:NMwCDT:C2}}\label{s-3}

In this section we prove Theorem~\ref{thm:NMwCDT:C2}.

\NMWCDTthm*

\begin{proof}
	It is readily seen that {\sc Node Multiway Cut with Deletable Terminals} belongs to \NP. We now reduce from {\sc Vertex Cover} on planar subcubic graphs, which is known to be \NP-complete~\cite{Mo01}. Let $G$ be the graph of an instance of this problem. We keep the same graph, but set $T = V(G)$. Since any two adjacent vertices are now adjacent terminals, any vertex cover in $G$ corresponds to a node multiway cut for $(G,T)$. The result follows.
\end{proof}

\section{Conclusions}

We proved that {\sc Edge Multiway Cut} and both versions of {\sc Node Multiway Cut} are \NP-complete for planar subcubic graphs.
We also showed that these results filled complexity gaps in the literature related to maximum degree, ${\cal H}$-topological-minor-free graphs and ${\cal H}$-subgraph-free graphs. 
The last dichotomy result assumes that ${\cal H}$ is a finite set of graphs. We therefore pose the following challenging question.

\begin{open}\label{o-1}
Classify the complexity of  {\sc Edge Multiway Cut} and both versions of {\sc Node Multiway Cut} for ${\cal H}$-subgraph-free graphs when ${\cal H}$ is infinite.
\end{open}

\noindent
An answer to Open Problem~\ref{o-1} will require novel insights into the structure of ${\cal H}$-subgraph-free graphs.

\bibliographystyle{abbrv}
\bibliography{arxiv}

\appendix

\section{The Proof of Lemma~\ref{lem:Mwc:C3}}

Here is the proof of Lemma~\ref{lem:Mwc:C3}, which is from Johnson et al.~\cite{JMOPPSV}, but which we include below for convenience.

\EMWClem*

\begin{proof}
Let $G$ belong to $\mathcal{H}$ and $T$ a set of terminals in $G$.	Let $G'$ be the graph $G$ after subdividing each edge. For each edge $e$ in $G$, there exist two edges in $G'$. If an edge of $G$ is in an edge multiway cut for $(G,T)$, then it suffices to replace it by only one of the two edges created from it in $G'$ to disconnect the path $e$ lies on. This yields an edge multiway cut for $(G',T)$ of the same size. Conversely, if an edge of $G'$ is in an edge multiway cut for $(G',T)$, then we replace it by the corresponding original edge of $G$. This yields an edge multiway cut for $(G,T)$ of the same size. Hence, $(G,T)$ has an edge multiway cut of size at most $k$ if and only if $(G',T)$ has an edge multiway cut of size~$k$. 
\end{proof}

\end{document}